\documentclass[11pt]{article}

\usepackage{fullpage}
\usepackage{amsmath,amsfonts,amsthm,cite}
\usepackage[unicode]{hyperref}
\usepackage{subfig}
\usepackage{svg}
\usepackage{xspace}
\usepackage{xargs}
\usepackage{graphicx}
\usepackage[labelfont=bf,font=small]{caption}
\usepackage[bottom]{footmisc}
\usepackage{tikz}
\usepackage{colortbl}
\pdfoutput=1

\graphicspath{{./Images/}}

\usepackage[colorinlistoftodos]{todonotes}

\newtheorem{theorem}{Theorem}[section]
\newtheorem{lemma}[theorem]{Lemma}
\newtheorem{corollary}[theorem]{Corollary}

\newcommandx{\Pull}{\textsc{Pull}\xspace}
\newcommandx{\Pullk}[2][1=$k$, 2=?]{\textsc{Pull{#2}}\text{-}#1\xspace}
\newcommandx{\PullkF}[2][1=$k$, 2=?]{\textsc{Pull{#2}}\text{-}#1\text{F}\xspace}
\newcommandx{\PullkFG}[2][1=$k$, 2=?]{\textsc{Pull{#2}}\text{-}#1\text{FG}\xspace}
\newcommandx{\PullkW}[2][1=$k$, 2=?]{\textsc{Pull{#2}}\text{-}#1\text{W}\xspace}
\newcommandx{\PullkWG}[2][1=$k$, 2=?]{\textsc{Pull{#2}}\text{-}#1\text{WG}\xspace}
\newcommandx{\PushPull}{\textsc{PushPull}\xspace}
\newcommandx{\PullPull}{\textsc{PullPull}\xspace}
\newcommand{\NCL}{\textsc{Nondeterministic Constraint Logic}\xspace}

\newcommand{\PTC}{\textsc{Planar 3-Coloring}\xspace}

\newcommand{\NP}{\textsc{NP}\xspace}
\newcommand{\PSPACE}{\textsc{PSPACE}\xspace}
\newcommand{\NPSPACE}{\textsc{NPSPACE}\xspace}
\newcommand{\YES}{\textsc{yes}\xspace}
\newcommand{\NO}{\textsc{no}\xspace}

\newcommand{\nlt}{nondeterministic locking 2-toggle\xspace}

{\makeatletter \gdef\fps@figure{!htbp}}

\let\realbfseries=\bfseries
\def\bfseries{\realbfseries\boldmath}

\def\emph#1{\textbf{\textit{\boldmath #1}}}

\newif\ifabstract
\abstracttrue
\newif\iffull
\ifabstract \fullfalse \else \fulltrue \fi

\ifabstract

\else

\fi

\newcounter{section-preserve}
\newcounter{theorem-preserve}
\newcommand{\blank}[1]{}
\newtoks\magicAppendix
\magicAppendix={}
\newtoks\magictoks
\newif\iflater
\laterfalse
\long\def\later#1{\iflater#1\else\magictoks={#1}%
	\edef\magictodo{\noexpand\magicAppendix={\the\magicAppendix \par
			\the\magictoks%
	}}
	\magictodo\fi}
\long\def\both#1{\iflater#1\else\magictoks={#1}%
	\edef\magictodo{\noexpand\magicAppendix={\the\magicAppendix \par
			\noexpand\setcounter{theorem-preserve}{\noexpand\arabic{theorem}}%
			\noexpand\setcounter{theorem}{\arabic{theorem}}%
			\noexpand\setcounter{section-preserve}{\noexpand\arabic{section}}%
			\noexpand\setcounter{section}{\arabic{section}}%
			\noexpand\let\noexpand\oldsection=\noexpand\thesection
			\noexpand\def\noexpand\thesection{\thesection}
			\noexpand\let\noexpand\oldlabel=\noexpand\label
			\noexpand\let\noexpand\label=\noexpand\blank
			\the\magictoks%
			\noexpand\setcounter{theorem}{\noexpand\arabic{theorem-preserve}}%
			\noexpand\setcounter{section}{\noexpand\arabic{section-preserve}}%
			\noexpand\let\noexpand\thesection=\noexpand\oldsection
			\noexpand\let\noexpand\label=\noexpand\oldlabel
	}}
	\magictodo
	\the\magictoks\fi}
\def\magicappendix{\latertrue \the\magicAppendix}

\title{\PSPACE-completeness of Pulling Blocks to Reach a Goal}
\author{%
  Joshua Ani%
    \thanks{Massachusetts Institute of Technology, Cambridge, MA, USA}
\and
  Sualeh Asif%
    \footnotemark[1]
\and
  Erik D. Demaine%
    \footnotemark[1]
\and
  Yevhenii Diomidov%
    \footnotemark[1]
\and
  Dylan Hendrickson%
    \footnotemark[1]
\and
  Jayson Lynch%
    \footnotemark[1]
\and
  Sarah Scheffler%
    \thanks{Boston University, Boston, MA, USA}
\and
  Adam Suhl%
    \thanks{Algorand, Boston, MA, USA}
}
\date{}

\begin{document}

\maketitle

\begin{abstract}
We prove \PSPACE-completeness of all but one problem in a large space of pulling-block problems where the goal is for the agent to reach a target destination. The problems are parameterized by whether pulling is optional, the number of blocks which can be pulled simultaneously, whether there are fixed blocks or thin walls, and whether there is gravity. We show \NP-hardness for the remaining problem, \PullkFG[1][?] (optional pulling, strength 1, fixed blocks, with gravity).
\end{abstract}

\section{Introduction}
\label{sec:intro}

In the broad field of \emph{motion planning}, we seek algorithms for actuating
or moving mobile agents (e.g., robots) to achieve certain goals.
In general settings, this problem is PSPACE-complete
\cite{Canny-1988-pspace,Reif-1979-mover},
but much attention has been given to finding simple variants near the threshold
between polynomial time and PSPACE-complete;
see, e.g., \cite{hearn2009games}.
One interesting and well-studied case, arising in warehouse maintenance,
is when a single robot with $O(1)$ degrees of freedom navigates an environment
with obstacles, some of which can be moved by the robot (but which cannot move
on their own).
Research in this direction was initiated in 1988 \cite{Wilfong-1991}.

A series of problems in this space arise from computer puzzle games,
where the robot is the agent controlled by the player,
and the movable obstacles are \emph{blocks}.
The earliest and most famous such puzzle game is \emph{Sokoban},
first released in 1982 \cite{sokoban-wiki}.
Much later, this game was proved PSPACE-complete
\cite{sokoban,hearn2009games}.
In Sokoban, the agent can \emph{push} movable $1 \times 1$ blocks
on a square grid, and the goal is to bring those blocks to target locations.
Later research in \emph{pushing-block puzzles} considered the simpler
goal of simply getting the robot to a target location,
proving various versions NP-hard, NP-complete, or PSPACE-complete
\cite{demainepush,demaine2003pushing,demaine2004pushpush}.

In this paper, we study the \Pull series of motion-planning problems
\cite{Ritt10,PRB16}, where the agent can \emph{pull} (instead of push)
movable $1 \times 1$ blocks on a square grid.
Figure~\ref{fig:example} shows a simple example.
This type of block-pulling mechanic (sometimes together with a block-pushing
mechanic) appears in many real-world video games,
such as Legend of Zelda, Tomb Raider, Portal, and Baba Is You.

\begin{figure}
\centering
\subfloat[Initial state]{\includegraphics[scale=1]{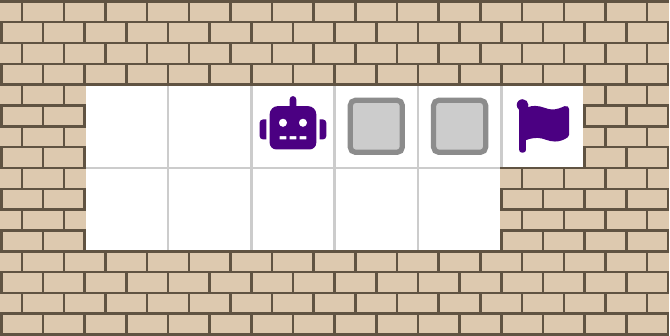}}
\hfill
\subfloat[A strength-1 move]{\includegraphics[scale=1]{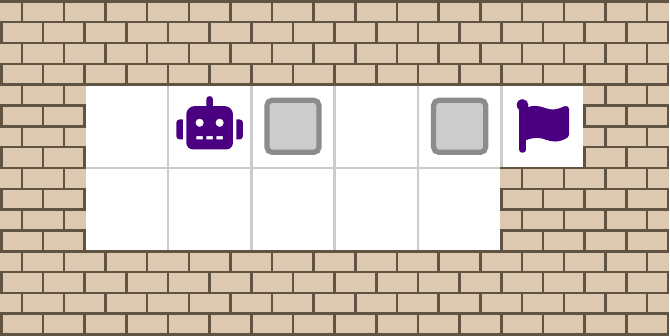}}

\caption{A pulling-block problem.  The robot is the agent, the flag is the goal square, the light gray blocks can be moved, and the bricks are fixed in place.
\sl \href{https://fontawesome.com/icons/robot?style=solid}{Robot} and \href{https://fontawesome.com/icons/flag?style=solid}{flag} icons from \href{https://fontawesome.com}{Font Awesome} under \href{https://creativecommons.org/licenses/by/4.0/}{CC BY 4.0 License}.}
\label{fig:example}
\end{figure}

We study several different variants of \Pull, which can be combined in arbitrary combination:
\begin{enumerate}
\setlength\itemsep{0pt}
\setlength\parskip{0pt}
\item \textbf{Optional/forced pulls:} In \textsc{Pull!}, every agent motion that can also pull blocks must pull as many as possible (as in many video games where the player input is just a direction). In \textsc{Pull?}, the agent can choose whether and how many blocks to pull. Only the latter has been studied in the literature, where it is traditionally called \textsc{Pull}; we use the explicit ``?''\ to indicate optionality and distinguish from \textsc{Pull!}.
\item \textbf{Strength:} In \Pullk[$k$][], the agent can pull an unbroken horizontal or vertical line of up to $k$ pullable blocks at once. In \Pullk[$\ast$][], the agent can pull any number of blocks at once.
\item \textbf{Fixed blocks/walls:} In \PullkF[][], the board may have fixed $1 \times 1$ blocks that cannot be traversed or pulled. In the \PullkW[][], the board may have fixed thin ($1 \times 0$) walls; this is more general because a square of thin walls is equivalent to a fixed block. Thin walls were introduced in \cite{demaine2017push}.
\item \textbf{Gravity:} In \textsc{Pull-G}, all movable blocks fall downward after each agent move. Gravity does not affect the agent's movement.
\end{enumerate}

Table~\ref{tab:results} summarizes our results: for all variants that include fixed blocks or
walls, we prove \PSPACE-completeness for any strength, with optional or forced pulls, and with or
without gravity, with the exception of \PullkFG[1][?] for which we only show \NP-hardness.

\definecolor{header}{rgb}{0.29,0,0.51}
\definecolor{gray}{rgb}{0.85,0.85,0.85}
\def\header#1{\multicolumn{1}{c}{\textcolor{white}{\textbf{#1}}}}
\def\tableref#1{[\S\ref{#1}]}

\begin{table}
    \centering
    \tabcolsep=0.5\tabcolsep
        \begin{tabular}{l  c c c c l l}
            \rowcolor{header}
            \header{Problem} & \header{\hspace{-0.3em}Forced} & \header{Strength\hspace{-0.2em}} & \header{Features} & \header{\hspace{-0.2em}Gravity\hspace{-0.2em}} & \header{Our result} & \header{\hspace{-0.3em}Previous best} \\
            \Pull?-$k$F & no & $k \ge 1$ & fixed blocks & no & \PSPACE-complete \tableref{sec:no gravity} & \NP-hard \cite{Ritt10} \\
            \rowcolor{gray}
            \Pull?-$\ast$F & no & $\infty$ & fixed blocks & no & \PSPACE-complete \tableref{sec:no gravity} & \NP-hard \cite{Ritt10} \\
            \Pull!-$k$F & yes & $k \ge 1$ & fixed blocks & no & \PSPACE-complete \tableref{sec:no gravity} &  \\
            \rowcolor{gray}
            \Pull!-$\ast$F & yes & $\infty$ & fixed blocks & no & \PSPACE-complete \tableref{sec:no gravity} &  \\
            \Pull?-1FG & no & $k = 1$ & fixed blocks & yes & \NP-hard \tableref{sec:Pull1FG NP} &  \\
            \rowcolor{gray}
            \Pull?-1WG & no & $k = 1$ & thin walls & yes & \PSPACE-complete \tableref{sec:optional pull} &  \\
            \Pull?-$k$FG & no & $k \ge 2$ & fixed blocks & yes & \PSPACE-complete \tableref{sec:optional pull} &  \\
            \rowcolor{gray}
            \Pull?-$\ast$FG & no & $\infty$ & fixed blocks & yes & \PSPACE-complete \tableref{sec:optional pull} &  \\
            \Pull!-$k$FG & yes & $k \ge 1$ & fixed blocks & yes & \PSPACE-complete \tableref{sec:mandatory gravity} &  \\
            \rowcolor{gray}
            \Pull!-$\ast$FG & yes & $\infty$ & fixed blocks & yes & \PSPACE-complete \tableref{sec:mandatory gravity} &  \\
        \end{tabular}
    \caption[]{Summary of our results.}
    \label{tab:results}
\end{table}

The only previously known hardness result for this family of problems is NP-hardness for both \PullkF and \PullkF[$*$] \cite{Ritt10}.
In some cases, our results are stronger than the best known results for the corresponding \textsc{Push} (pushing-block) problem; see \cite{PRB16}.
More complex variants \PullPull (where pulled blocks slide maximally), \PushPull (where blocks can be pushed and pulled), and \textsc{Storage Pull} (where the goal is to place multiple blocks into desired locations) are also known to be PSPACE-complete \cite{demaine2017push,PRB16}.

Our reductions are from Asynchronous Nondeterministic Constraint Logic (NCL)
\cite{hearn2009games, DBLP:conf/cccg/Viglietta13} and
planar 1-player motion planning \cite{demaine2018general, doors}.
In Section~\ref{sec:no gravity}, we reduce from NCL to prove \PSPACE-hardness of all nongravity variants.
In Section~\ref{sec:gravity pspace}, we use the motion-planning-through-gadgets framework \cite{demaine2018general} to prove \PSPACE-completeness of most variants with gravity, including all variants with forced pulling and variants with optional pulling and either thin walls or fixed blocks with $k\ge2$.
These reductions use two particular gadgets for 1-player motion planning,
the newly introduced \emph{\nlt}
(a variant of the locking 2-toggle from \cite{demaine2018general})
and the \emph{3-port self-closing door} (one of the self-closing doors from
\cite{doors}).
Although the latter gadget is proved hard in \cite{doors}, for completeness,
we give a more succinct proof in Appendix~\ref{app:self-closing door}.
In Section~\ref{sec:Pull1FG NP}, we prove \NP-hardness for the one remaining case of \PullkFG[1][?],
again reducing from 1-player planar motion planning,
this time with an NP-hard gadget called the crossing NAND gadget \cite{doors}.

\section{Pulling Blocks with Fixed Blocks is \PSPACE-complete}
\label{sec:no gravity}

In this section, we show the PSPACE-completeness of all variants of pulling-block problems we have defined without gravity, namely \PullkF, \PullkW, \PullkF[$k$][!], and \PullkW[$k$][!] for $k \ge 1$, and \PullkF[$\ast$], \PullkW[$\ast$], \PullkF[$*$][!], and \PullkW[$*$][!].  We do this through a reduction from Nondeterministic Constraint Logic \cite{hearn2009games}, which we describe briefly before moving on to the main proof.

\subsection{Asynchronous Nondeterministic Constraint Logic}
\label{ssec:NCL}
\NCL (NCL) takes place on \emph{constraint graphs}: a directed graph where each edge has weight 1 or 2.  Weight-1 edges are called \emph{red}; weight-2 edges are called \emph{blue}.  The ``constraint'' in NCL is that each vertex must maintain in-weight at least 2.  A \emph{move} in NCL is a reversal of the direction of one edge, while maintaining compliance with the constraint.  

In \emph{asynchronous} NCL, the process of switching the orientation of an edge does not happen instantaneously, but instead it takes a positive amount of time, and it is possible to be in the process of switching several edges simultaneously. When an edge is in the process of being reversed, it is not oriented towards either vertex.  Viglietta \cite{DBLP:conf/cccg/Viglietta13} showed that this model is equivalent to the regular (synchronous) model, because there is no benefit to having an edge in the intermediate unoriented state.  In this work, we only use the asynchronous NCL model; any mention of NCL should be understood to mean asynchronous NCL.

An instance of \NCL consists of a constraint graph $G$ and an edge $e$ of $G$, called the \emph{target edge}.  
The output is \YES if there is a sequence of moves on $G$ that reverses the direction of $e$, and \NO otherwise.  \NCL is \PSPACE-complete, even for planar constraint graphs that have only two types of vertices: AND (two red edges, one blue edge) and OR (three blue edges).  
We will reduce from the planar, AND/OR, asynchronous version of NCL to show pulling-block problems without gravity \PSPACE-hard.
For more description of NCL, including a proof of \PSPACE-completeness, the reader is referred to \cite{hearn2009games}.

\subsection{NCL Gadgets in Pulling Blocks}

In order to embed an NCL constraint graph into \PullkF, we need three
components, corresponding to NCL edges (which can attach to AND and OR gadgets
in all necessary orientations, and that allows the player to win if the
winning edge is flipped), AND vertices, and OR vertices.
In each of these gadgets, we will show that if the underlying NCL constraint is violated, then the agent will be ``trapped'', meaning that the state is in an \emph{unrecoverable configuration}, a concept used in several previous blocks games \cite{sokoban,hearn2009games}.  This occurs when the agent makes a pull move after which no set of moves will lead to a solution, generally because the agent has trapped itself in a way that no pull can be made \textit{at all} (or only a few more pull moves may be made, and all of them lead to a state such that there are no more pull moves).

\textbf{Diode Gadget.} Before describing the three main gadgets, we describe a helper gadget, the \emph{diode}, shown in Figure~\ref{fig:diode}. The diode can be repeatedly traversed in one direction but never the other. It was introduced in \cite{Ritt10}.

\begin{figure}
\centering
$\vcenter{\hbox{\includegraphics[scale=0.6]{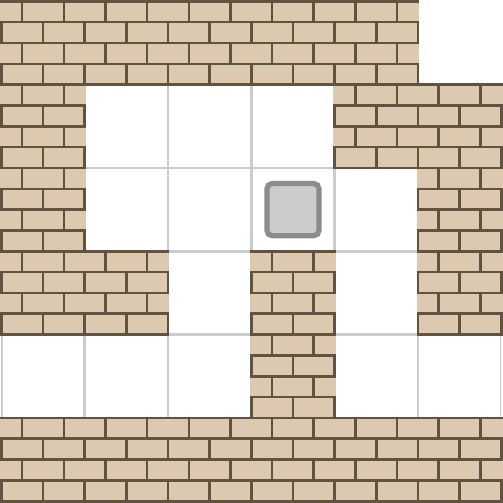}}}
\quad \equiv \quad
\vcenter{\hbox{\includegraphics[angle=90]{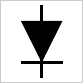}}}$
\caption{Diode gadget, which can be repeatedly traversed from left to right but never from right to left.
In diagrams to follow it will be represented by the diode symbol.}
\label{fig:diode}
\end{figure}

In the next three sections, we describe the three main gadgets in turn.

\subsubsection{Edge Gadget}

\begin{figure}
\centering
\begin{tikzpicture}
	\node at (0,0) {\includegraphics[width=.9\textwidth]{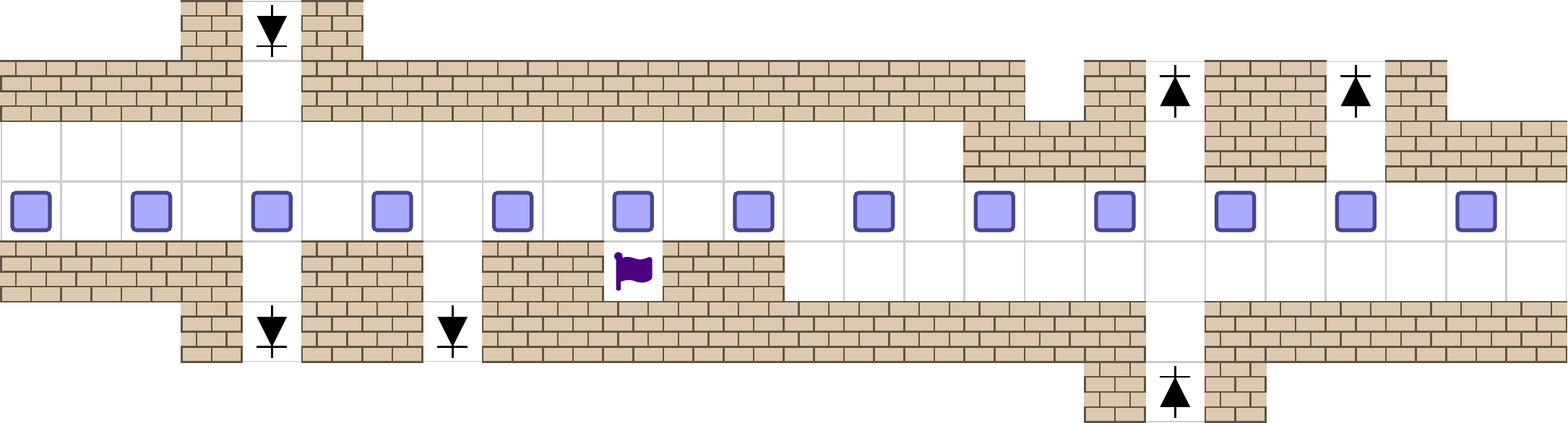}};
	\draw[->, line width=0.8] (-7.2,0) -- (-6.57,0);
	\draw[->, line width=0.8] (-6.57,0) -- (-7.2,0);

	\draw[->, line width=0.8] (7.2,0) -- (6.57,0);
	\draw[->, line width=0.8] (6.57,0) -- (7.2,0);
\end{tikzpicture}
\caption{Edge gadget for building NCL constraint graph in \PullkF.  This edge encodes an NCL wire pointing to the right (opposing the blocks, which are moved to the left).  This figure shows the winning edge, which if flipped allows reaching the goal; nonwinning edges are the same but without the flag.}
\label{fig:wire}
\end{figure}
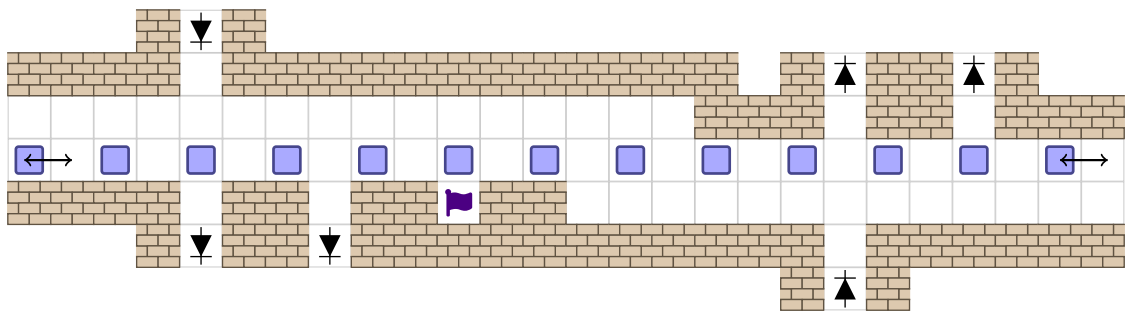

The edge gadget, shown in Figure~\ref{fig:wire}, encodes a single edge (of either weight) from NCL into \PullkF.   The blocks can shift by exactly one space; whether they are moved left or right (or up or down) corresponds to the NCL edge pointing right or left (or down or up) respectively.  (Note that the orientation of the NCL edge \emph{opposes} the direction of the blocks.)  The ends of the wires will be in vertex gadgets, which are explained below.

The diodes on the sides are to allow the agent to traverse between edges without going through a vertex gadget.  The position and orientation of the diode gadgets prevents the agent from pulling a block out of the edge gadget without trapping itself.

The edge shown in Figure~\ref{fig:wire} only goes in a straight line; it may turn corners via the corner gadget in Figure~\ref{fig:corner}.  We fix small misalignment of wires at the gadgets, not on the wires.

\begin{figure}
\centering
\includegraphics[scale=0.6]{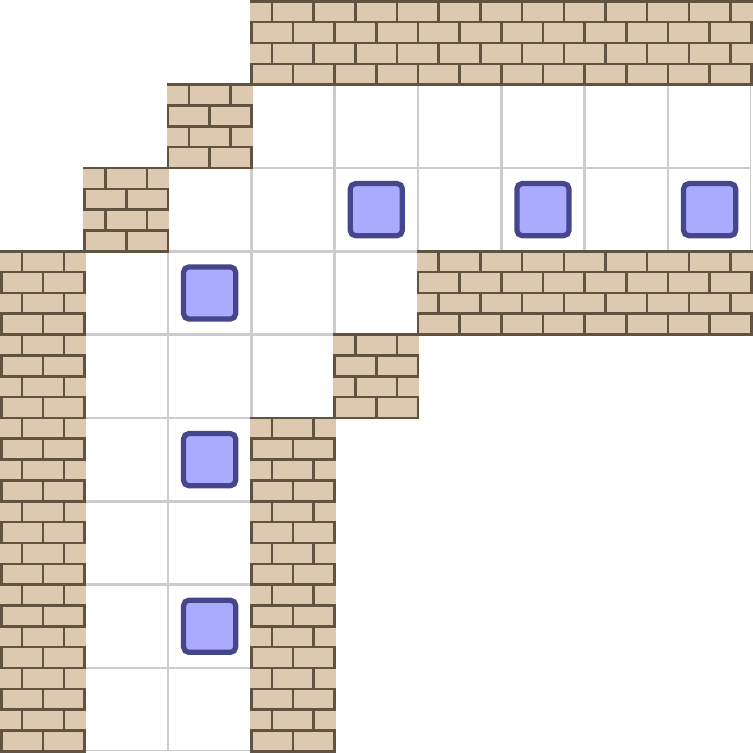}
\caption{Gadget for allowing wires to turn corners.  Currently oriented with the NCL wire pointed left/down.}
\label{fig:corner}
\end{figure}

This wire gadget also allows encoding the win condition in the target edge: the finish tile can be in a small room coming out of the wall, blocked by the blocks' current placement, as shown in Figure~\ref{fig:wire}; it can then be reached only if the edge can be reversed.
(Without the flag, Figure~\ref{fig:wire} shows an ordinary edge gadget.)

Notice that the only location the agent can ever try to cheat and pull a block out of the wire is by pulling a block up into the downward-facing diode at the top-left of the gadget, or (symmetrically) by pulling a block down into the upward-facing diode at the bottom-right of the gadget.  If the agent does this, then the game is now in an unrecoverable configuration---the agent cannot escape back into the wire, because the block it just pulled is blocking the way.  But it also cannot traverse the diode the wrong way to escape in the other direction.  Thus, the agent cannot pull blocks without reaching an unrecoverable configuration, except for the moves which correspond to reversing the NCL edge.

The player may partially reverse an NCL edge and exit before completing the reversal. This leaves a gap of two empty squares between consecutive movable blocks somewhere in the edge gadget. This is why we reduce from asynchronous NCL; while in this partially reversed state, the NCL edge is not oriented towards either vertex, and each vertex gadget behaves as though the edge were oriented away from it.

\subsubsection{OR gadget}
\label{sssec:or}

\begin{figure}
\centering
\begin{tikzpicture}
		\node at (0,0) {\includegraphics[width=9cm]{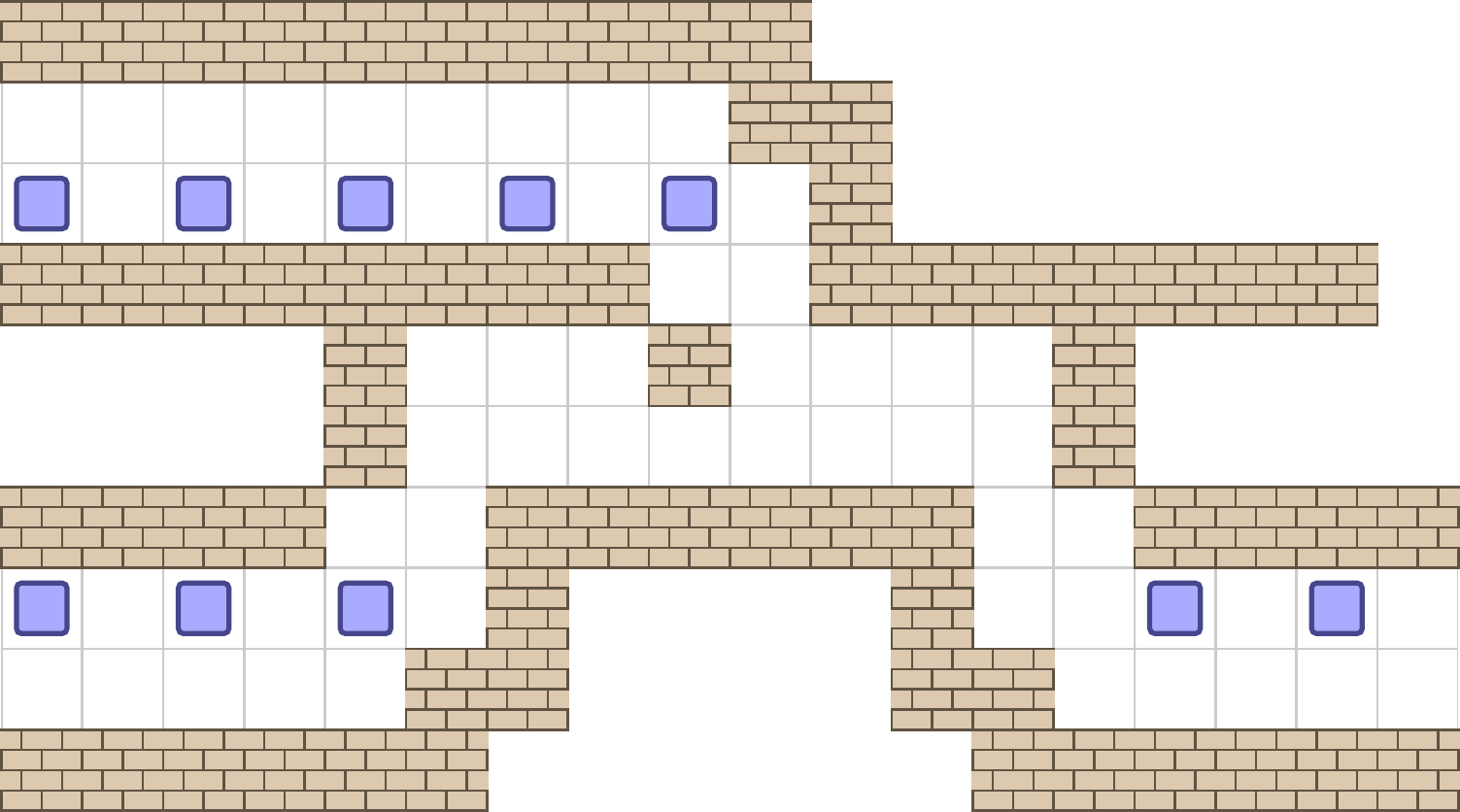}};
        \draw[->, line width=0.8] (-0.25,1.25) -- (-0.75,1.25);
        \draw[->, line width=0.8] (-0.75,1.25) -- (-0.25,1.25);
        \draw[->, line width=0.8] (2.75,-1.25) -- (2.25,-1.25);
        \draw[->, line width=0.8] (2.25,-1.25) -- (2.75,-1.25);
        \draw[->, line width=0.8] (-2.25,-1.25) -- (-2.75,-1.25);
        \draw[->, line width=0.8] (-2.75,-1.25) -- (-2.25,-1.25);
\end{tikzpicture}
\caption{Gadget for an NCL OR vertex.  Currently, the left edge points out, the right edge points in, and the top edge points out.}
\label{fig:or}
\end{figure}
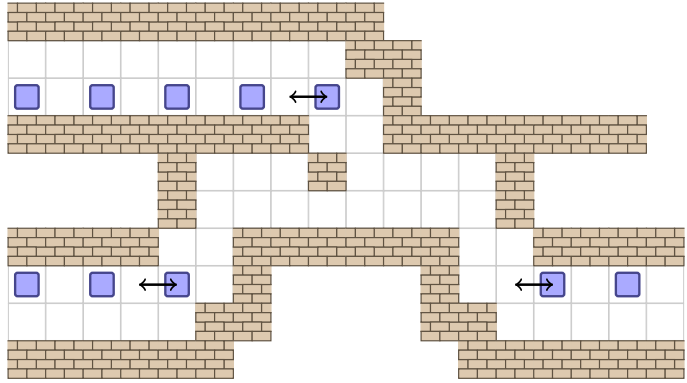

The OR gadget, shown in Figure~\ref{fig:or}, consists of an area fully enclosed by walls except at three connections to edge gadgets.
The agent can enter and exit the enclosed area through an edge connection when the blocks in the edge are pulled away from the OR gadget (i.e., when the NCL edge points in).
When the edge blocks are pulled inward (i.e., NCL edge points out), the agent cannot escape the enclosed area through that edge gadget.
Thus, when inside the enclosed area, the agent may pull an edge block in (i.e., start switching the NCL edge to point out), but if both other NCL edges already point out, the agent will be trapped inside the gadget. This enforces the constraint that at least one edge must point towards the OR vertex.

This gadget can vary in size: if some edge gadgets are slightly misaligned, the middle part can be made bigger or smaller to accommodate---the interior of the gadget needs only to be fully enclosed except at the incident edge gadgets.

\begin{lemma} \label{lem:NCL-vertices-OR}
  The OR gadget enforces exactly the constraints of an NCL OR vertex.
\end{lemma}

\begin{proof}
The NCL OR constraint is that at least one edge must point into the vertex at all times. If the constraint is satisfied, then at least one wire of blocks is pushed out, and the agent can escape through that edge gadget.  If the agent tries to violate the constraint, then the first move it must take in order to make the last edge point away from the vertex is to pull the last block of the corresponding edge gadget inward.  This puts the gadget into an unrecoverable condition: the agent is now trapped in the OR gadget.
\end{proof}

\subsubsection{AND gadget}
\label{sssec:and}

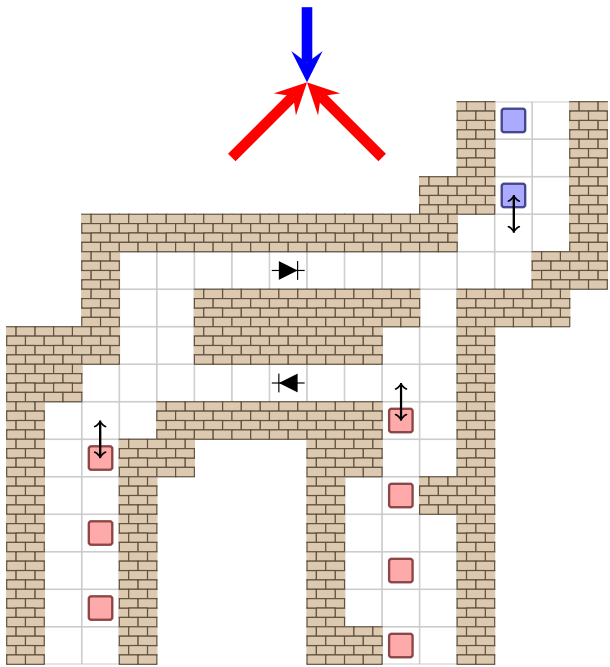
\begin{figure}
\centering
\begin{tikzpicture}
		\node at (0,0) {\includegraphics[width=8cm]{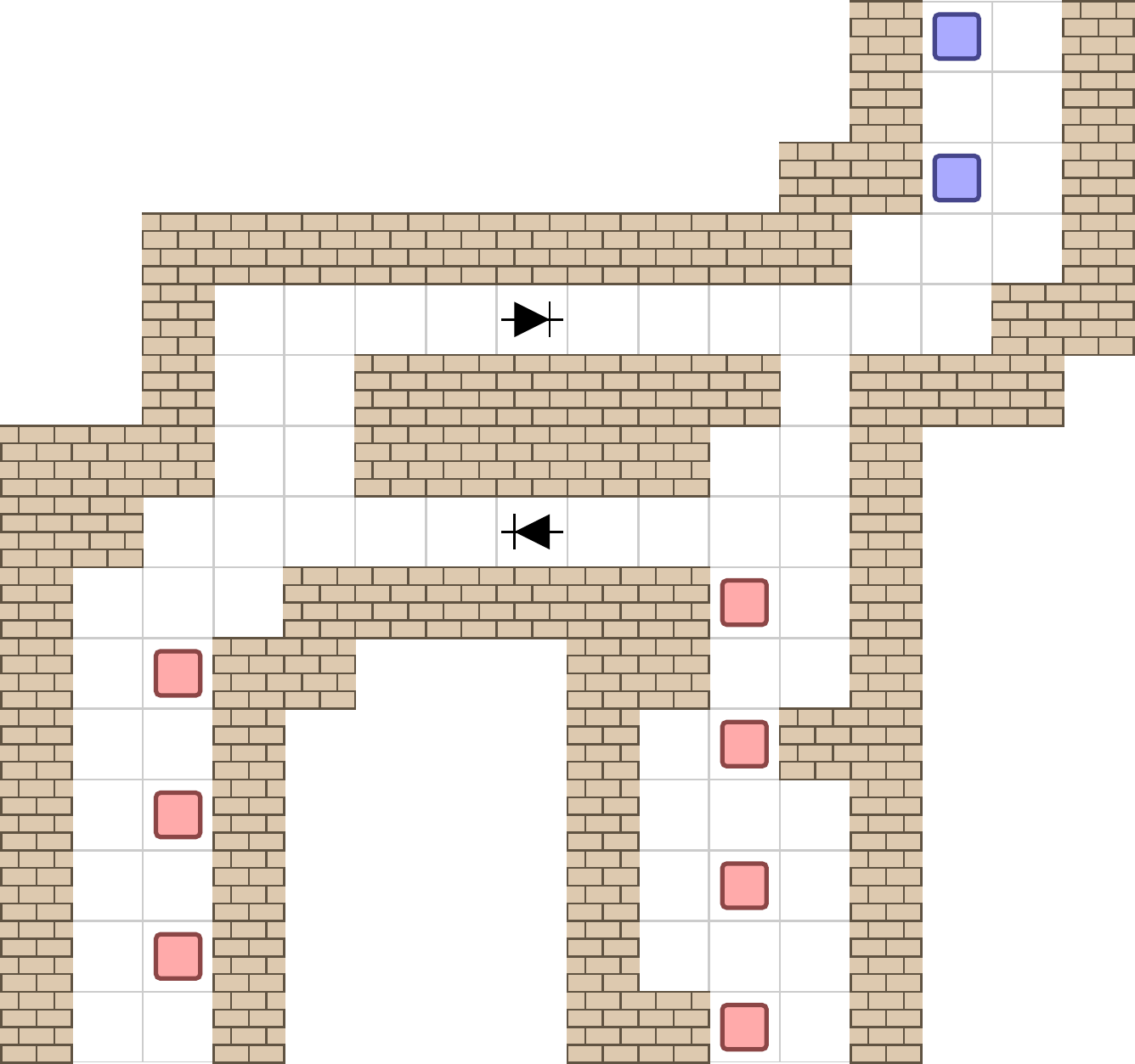}};
        \draw[->, line width=0.8] (2.75,2) -- (2.75,2.5);
        \draw[->, line width=0.8] (2.75,2.5) -- (2.75,2);
        \draw[->, line width=0.8] (1.25,-0.5) -- (1.25,0);
        \draw[->, line width=0.8] (1.25,0) -- (1.25,-0.5);
        \draw[->, line width=0.8] (-2.75,-1) -- (-2.75,-0.5);
        \draw[->, line width=0.8] (-2.75,-0.5) -- (-2.75,-1);

		\draw[->, line width=4, >=stealth, color=red] (-1,3) -- (0,4);
		\draw[->, line width=4, >=stealth, color=red] (1,3) -- (0,4);
		\draw[->, line width=4, >=stealth, color=blue] (0,5) -- (0,4);
\end{tikzpicture}
\caption{Gadget for an NCL AND vertex, currently with all three edges pointing in.  The lower diode allows traversal from right to left, and is blocked if the right red edge is pointing away; the upper diode allows traversal from left to right.  If the (top) blue edge is pointing in, the agent can escape through that edge gadget.  To escape the AND gadget after making the blue edge point away by pulling the bottommost blue block down, both red edges must be pointing in, so that the agent can go through the bottom diode and escape through the left red edge.}
\label{fig:and}
\end{figure}

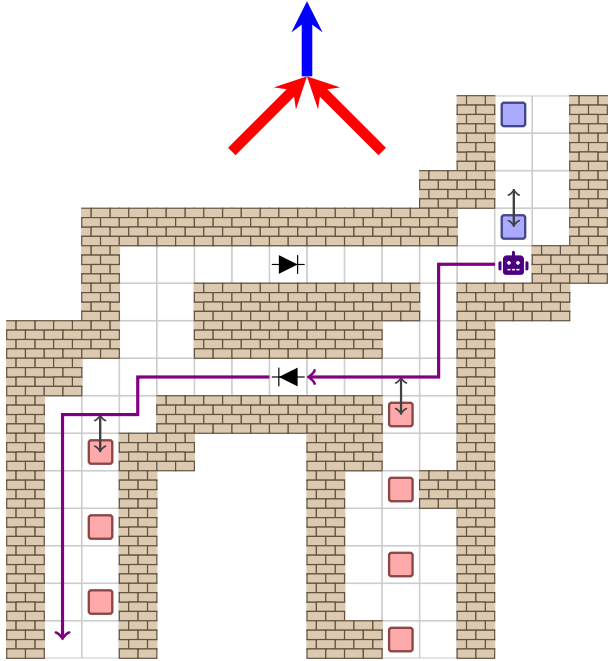
\begin{figure}
\centering
\begin{tikzpicture}
		\node at (0,0) {\includegraphics[width=8cm]{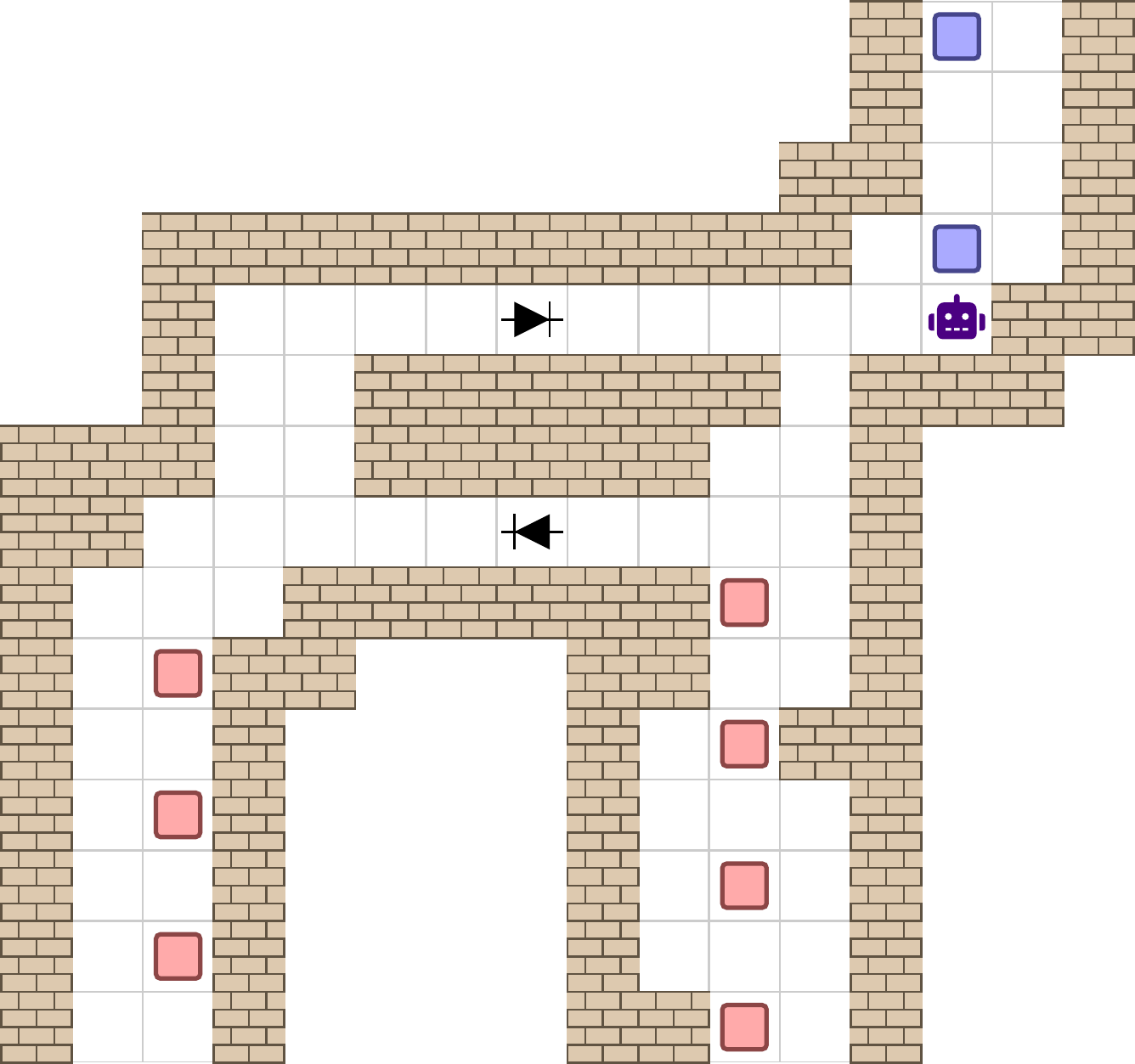}};
        \draw[->, line width=0.8, color=darkgray] (2.75,2) -- (2.75,2.5);
        \draw[->, line width=0.8, color=darkgray] (2.75,2.5) -- (2.75,2);
        \draw[->, line width=0.8, color=darkgray] (1.25,-0.5) -- (1.25,0);
        \draw[->, line width=0.8, color=darkgray] (1.25,0) -- (1.25,-0.5);
        \draw[->, line width=0.8, color=darkgray] (-2.75,-1) -- (-2.75,-0.5);
        \draw[->, line width=0.8, color=darkgray] (-2.75,-0.5) -- (-2.75,-1);

		\draw[->, line width=1.2, color=violet] (2.5,1.5) -- (1.75,1.5) -- (1.75,0) -- (0,0);
		\draw[->, line width=1.2, color=violet] (-.5,0) -- (-2.25,0) -- (-2.25,-.5) -- (-3.25,-.5) -- (-3.25,-3.5);

		\draw[->, line width=4, >=stealth, color=red] (-1,3) -- (0,4);
		\draw[->, line width=4, >=stealth, color=red] (1,3) -- (0,4);
		\draw[<-, line width=4, >=stealth, color=blue] (0,5) -- (0,4);
\end{tikzpicture}
\caption{The agent has just started to flip the blue NCL edge outward by pulling a blue block inward. Both red NCL edges are pointed inward, so the agent can traverse the lower diode and escape out the left red edge. Note that if either red NCL edge were pointed outward, escape would be impossible.} 
\label{fig:and2}
\end{figure}

\begin{figure}
\centering
\begin{tikzpicture}
		\node at (0,0) {\includegraphics[width=8cm]{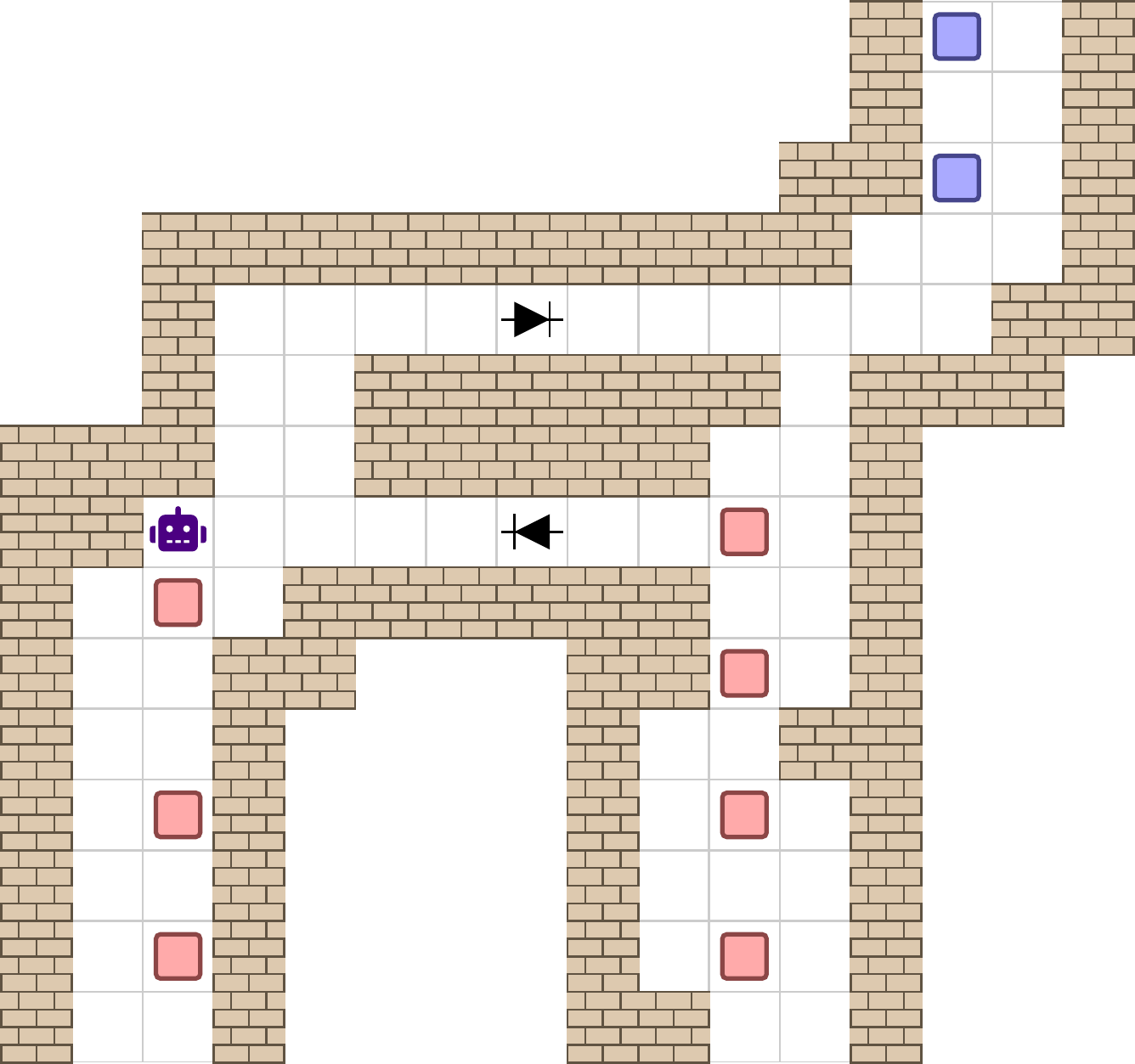}};
        \draw[->, line width=0.8, color=darkgray] (2.75,2) -- (2.75,2.5);
        \draw[->, line width=0.8, color=darkgray] (2.75,2.5) -- (2.75,2);
        \draw[->, line width=0.8, color=darkgray] (1.25,-0.5) -- (1.25,0);
        \draw[->, line width=0.8, color=darkgray] (1.25,0) -- (1.25,-0.5);
        \draw[->, line width=0.8, color=darkgray] (-2.75,-1) -- (-2.75,-0.5);
        \draw[->, line width=0.8, color=darkgray] (-2.75,-0.5) -- (-2.75,-1);

		\draw[->, line width=1.2, color=violet] (-2.5,0) -- (-1.75,0) -- (-1.75,1.5) -- (-.5,1.5);
		\draw[->, line width=1.2, color=violet] (0,1.5) -- (2.25,1.5) -- (2.25,2) -- (3.25,2) -- (3.25,3.5);

		\draw[<-, line width=4, >=stealth, color=red] (-1,3) -- (0,4);
		\draw[<-, line width=4, >=stealth, color=red] (1,3) -- (0,4);
		\draw[->, line width=4, >=stealth, color=blue] (0,5) -- (0,4);
\end{tikzpicture}
\caption{The agent has just started to flip the left red NCL edge outward by pulling the red block inward. Because the blue NCL edge points inward, the agent can traverse the top diode and escape out the blue edge.} 
\label{fig:and3}
\end{figure}
An NCL AND vertex has two red (weight 1) edges and one blue (weight 2) edge. Its constraint is that the blue edge may point outward only if both red edges point inward.  Our AND vertex gadget in \PullkF is shown in Figure~\ref{fig:and}.

Like the OR gadget, the AND gadget traps the agent inside the gadget if the agent tries to violate the NCL constraint. Two of the edge connections, one red and one blue, are like those of the OR gadget, allowing the agent to escape the gadget into the edge if the blocks have been pulled outward (i.e., if the NCL edge points inward).
The remaining red edge connection is different: the agent can never escape into this edge. Instead, when this edge's blocks are pulled outward (i.e., when the NCL edge points inward), it unblocks a path allowing the agent to traverse from the blue-exit side of the gadget to the red-exit side of the gadget.

An agent inside the gadget trying to pull the blue edge block inward (i.e., start switching the blue NCL edge to point outward) is trapped on the blue-exit side of the gadget unless this special red edge has its blocks pulled outward (i.e., the red NCL edge points inward); even then, the agent is still trapped inside the gadget unless the other red edge also has its blocks pulled outward to allow escape (i.e., the other red NCL edge also points inward). Thus an agent trying to switch the blue NCL edge outward is trapped unless both red NCL edges point inward, enforcing the AND condition. This is illustrated in Figure~\ref{fig:and2}

The AND gadget contains two reusable one-way gadgets. The lower diode is blocked if the right red edge points away, trapping the agent if the blue edge also points away, but allowing the agent to traverse from right to left and escape if both red edges point in.  The upper diode allows the agent to travel from the red-exit side to the blue-exit side regardless of the state of the third edge; this is necessary to ensure the agent can escape out the blue exit (if the blue edge points in) after flipping either red edge to point away, as illustrated in Figure~\ref{fig:and3}.

As with the OR gadget, if the incident edge gadgets are aligned with different parity, this gadget can be expanded slightly accommodate the edge gadgets.

\begin{lemma} \label{lem:NCL-vertices-AND}
  The AND gadget enforces exactly the constraints of an NCL AND vertex.
\end{lemma}

\begin{proof}
The NCL AND constraint is that either the blue edge (top) or both red edge (bottom left and right) point towards the vertex.
If both red edges are pointing in, then the agent can pull the bottom blue block in and then escape through the left red edge gadget by going through the bottom diode.
If the blue edge is pointing in, then the agent can always escape through the blue edge gadget. The agent can orient the left red edge out by entering through the red-exit, going through the top diode, and leaving through the blue exit; and it can orient the right red edge out by entering and leaving through the blue exit.

The agent can attempt to violate the constraint by making both the blue edge and at least one red edge point away from the AND vertex. We consider the two red edges separately.
First, if both the left red edge and the blue edge point out, then both exit points from the gadget are blocked, so the agent is trapped.
Second, suppose that the constraint becomes violated by both the blue edge and the right red edge pointing away. Then the agent has just pulled either the bottom blue block on the top red block into the AND gadget, and the agent is in the right side of the gadget. The agent is trapped: the blue exit is blocked by the blue edge pointing away, the bottom diode is blocked by the red edge pointing away, and the top diode cannot be traversed from right to left.
\end{proof}

\subsection{Proof of \PSPACE-completeness}

We first observe that every pulling-block problem we consider is in \PSPACE.

\begin{lemma}
\label{lem:pull-in-PSPACE}
Every pulling-block problem defined in Section~\ref{sec:intro} is in \PSPACE.
\end{lemma}

\begin{proof}
The entire configuration while playing on instance of a pulling-block problem can be stored in polynomial space (e.g., as a matrix recording whether each cell is empty, a fixed block, a movable block, the agent's location, or the finish tile). There is a simple nondeterministic algorithm which guesses each move and keeps track of the configuration using only polynomial space, accepting if the agent reaches the goal square.
Thus the problem is in \NPSPACE, so by Savitch's Theorem \cite{savitch1970relationships} it is also in \PSPACE.
\end{proof}

\begin{theorem}
\label{thm:pull-kF-PSPACE-complete}
\PullkF and \PullkF[$k$][!] PSPACE-complete for $k \ge 1$ and $k=*$.
\end{theorem}

\begin{proof}
Lemma~\ref{lem:pull-in-PSPACE} gives us containment in \PSPACE.
For \PSPACE-hardness, we reduce from asynchronous NCL
(as defined in Section~\ref{ssec:NCL}).

Given a planar AND/OR NCL graph, we construct an instance of \PullkF or \PullkF[$k$][!] as follows. First, embed the graph in a grid graph. Scale this grid graph by enough to fit our gadgets; $20\times20$ suffices. At each vertex, place the appropriate AND or OR vertex gadget. Place edge gadgets in the appropriate configuration along each edge, using corner gadgets on turns. Adjust the vertex gadgets to accommodate the alignment of the edge gadgets incident to them. Finally, place the goal tile in the edge gadget corresponding to the target edge so that it is accessible only if the target edge is flipped, and place the agent on any empty tile. 

The agent can walk through edge gadgets to visit any NCL edge or vertex, and by
Lemmas~\ref{lem:NCL-vertices-OR} and~\ref{lem:NCL-vertices-AND}, flip edges in accordance with the rules of NCL. Ultimately, it can reach the goal tile if and only if the target edge of the NCL instance can be reversed.

In our construction, the agent never has the opportunity to pull more than 1 block at a time.  Thus the reduction works for \PullkF for any $k\geq1$, including $k=*$. In addition, the agent never has to choose not to pull a block when taking a step, so the reduction works for \PullkF[$k$][!] as well as \PullkF.
\end{proof}

\begin{corollary}
\label{cor:pull-W-PSPACE-complete}
\PullkW and \PullkW[$k$][!] are \PSPACE-complete for $k\geq1$ and $k=*$.
\end{corollary}

\begin{proof}
A fixed block can be simulated using four thin walls drawn around a single tile, so our constructions can be built using thin walls instead of fixed blocks. Formally, this is a reduction from \PullkF to \PullkW and a reduction from \PullkF[$k$][!] to \PullkW[$k$][!].
\end{proof}

\section{\texorpdfstring{\PullkFG}{Pull?-kFG} is \PSPACE-complete for \texorpdfstring{$k \ge 2$}{k ≥ 2} and \texorpdfstring{\PullkFG[$k$][!]}{Pull!-kFG} is \PSPACE-complete for \texorpdfstring{$k \ge 1$}{k ≥ 1}}
\label{sec:gravity pspace}

In this section, we show \PSPACE-completeness results for most of the
pulling-block variants with gravity.
In Section~\ref{sec:gadgets}, we introduce and prove results about
\emph{1-player motion planning} from the motion-planning-through-gadgets
framework introduced in \cite{demaine2018computational}, which will be the
basis for the later proofs.
In Section~\ref{sec:optional pull},
we show \PSPACE-completeness for \PullkFG with $k \ge 2$, for
\PullkFG[$\ast$], for \PullkWG with $k \ge 1$, and for \PullkWG[$\ast$].
In Section~\ref{sec:mandatory gravity}, we show \PSPACE-completeness for
\PullkFG[$k$][!] with $k \ge 1$, and for \PullkFG[$\ast$][!].
The one case missing from this collection is \PullkFG[1], which we prove
NP-hard later in Section~\ref{sec:Pull1FG NP}.

\subsection{1-player Motion Planning}
\label{sec:gadgets}

\emph{1-player motion planning} refers to the general problem of planning an agent's motion to complete a path through a series of gadgets whose state and traversability can change when the agent interacts with them. In particular, a \emph{gadget} is a constant-size set of locations, states, and traversals, where each traversal indicates that the agent can move from one location to another while changing the state of the gadget from one state to another. A system of gadgets is constructed by connecting the locations of several gadgets with a graph, which is sometimes restricted to be planar. The decision problem for 1-player motion planning is whether the agent, starting from a specified stating location, can follow edges in the graph and transitions within gadgets to reach some goal location.

Our results use that 1-player planar motion planning is \PSPACE-complete
for the following gadgets:
\begin{enumerate}
\item The \emph{locking 2-toggle}, shown in Figure~\ref{fig:l2t}, is a three-state two-tunnel reversible deterministic gadget. In the \emph{middle state}, both tunnels can be traversed in one direction, switching to one of two \emph{leaf states}. Each leaf state only allows the transition back across that tunnel in the opposite direction, returning the gadget to the middle state. Traversing one tunnel ``locks'' the other side from being used until the prior traversal is reversed.

1-player planar motion planning with locking 2-toggles was shown
\PSPACE-complete in \cite{demaine2018general}.
In Section~\ref{sec:leaf2toggle}, we strengthen the result in \cite{demaine2018general} by showing that 1-player motion planning with locking 2-toggle remains hard even if the initial configuration of the system has all gadgets in leaf (locked) states.

 \begin{figure}[ht]
 	\begin{minipage}[b]{0.45\linewidth}
 		\centering
	\includegraphics[width=0.45\textwidth]{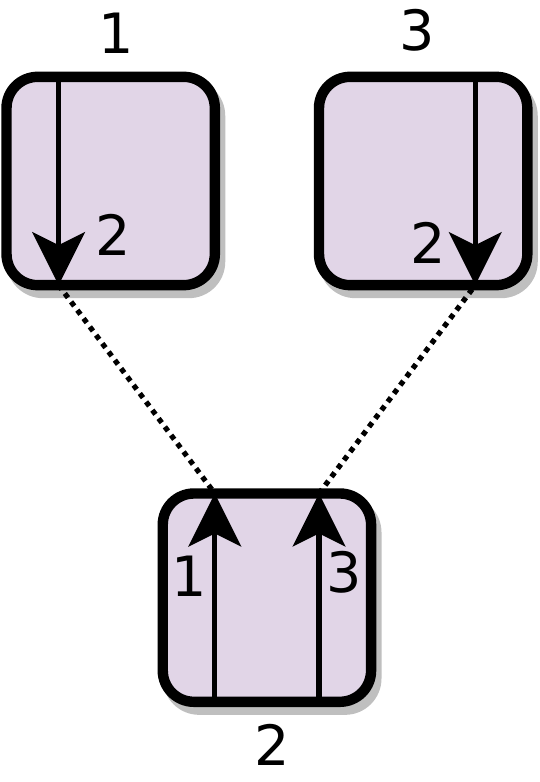}
	\caption{State space of the locking 2-toggle.}
	\label{fig:l2t}
 	\end{minipage}
 	\hspace{0.5cm}
 	\begin{minipage}[b]{0.5\linewidth}
 		\centering
	\includegraphics[width=0.4\textwidth]{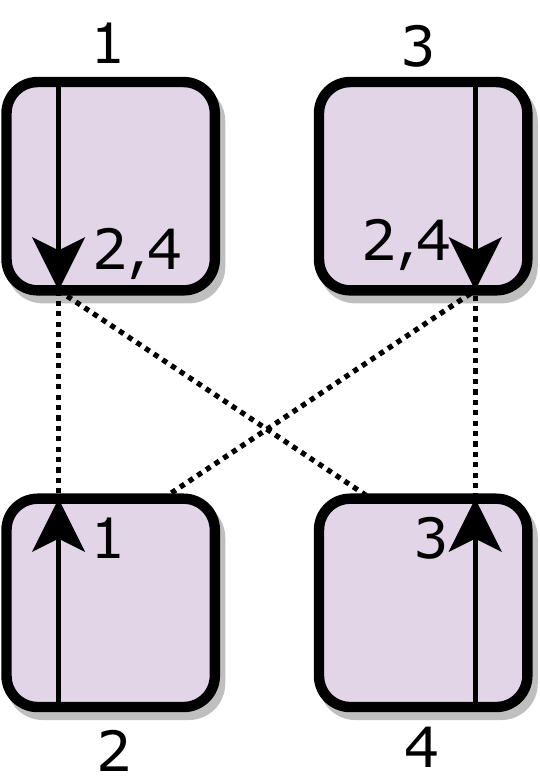}
	\caption{State space of the nondeterministic locking 2-toggle.}
	\label{fig:nl2t}
 	\end{minipage}
 \end{figure}
 
\item The \emph{\nlt}, shown in Figure~\ref{fig:nl2t}, is a four-state gadget where each state has two transitions, each across the same tunnel. The top pair of states each allow a single traversal downward, and allow the agent to choose either of the two bottom states for the gadget. Similarly, the bottom pair of states each allow a single traversal upward to one of the top states. We can imagine this as being similar to the locking 2-toggle if the tunnel to be taken next is guessed ahead of time: the bottom state of the locking 2-toggle is split into two states which together allow the same traversals, but only if the agent picks the correct one ahead of time.

In Section~\ref{sec:nondet2toggle}, we show that 1-player motion planning with the \nlt is \PSPACE-complete.

\item The \emph{door gadget} has three directed tunnels called \emph{open}, \emph{close}, and \emph{traverse}. The traverse tunnel is open or closed depending on the state of the gadget and does not change the state. Traversing the open or close tunnel opens or closes the traverse tunnel, respectively.

1-player motion planning with door gadgets was shown
\PSPACE-complete in \cite{nintendoor} and explored more thoroughly
(in particular, proved hard for most planar cases) in \cite{doors}.

\item The \emph{3-port self-closing door}, shown in Figure~\ref{fig:statespace-scd}, is a gadget with a tunnel that becomes closed when the agent traverses it and a location that the agent can visit to reopen the tunnel. 
It has an \emph{opening port}, which opens the gadget,
and a \emph{self-closing tunnel}, which is the tunnel that closes when traversed.

In Appendix~\ref{app:self-closing door}, we prove that 1-player planar motion planning with the \emph{3-port self-closing door} is \PSPACE-complete.
A more general result on self-closing doors can be found in \cite{doors}, but we include this more succinct proof for completeness and conciseness. 

\begin{figure}
\centering
\includegraphics[width=0.3\textwidth]{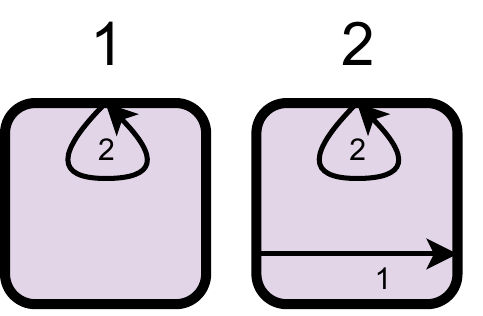}
\caption{State space of the 3-port self-closing door, used in the \PullkFG[$k$][!] reduction.}
\label{fig:statespace-scd}
\end{figure} 
 
\end{enumerate}

\subsubsection{Nondeterministic Locking 2-toggle}
\label{sec:nondet2toggle}
\label{sec:leaf2toggle}

In this section, we prove that 1-player motion planning with the \nlt is \PSPACE-complete. We also show that 1-player motion planning with the locking 2-toggle remains \PSPACE when the gadgets are restricted to start in leaf states.

We use the construction shown in Figure~\ref{fig:nl2t-to-l2t} to show simultaneously that locking 2-toggles starting in leaf states can simulate a locking 2-toggle starting in a nonleaf state, and nondeterministic locking 2-toggles can simulate a locking 2-toggle. This construction consists of two nondeterministic locking 2-toggles and a 1-toggle. A \emph{1-toggle} is a two-state, two-location, reversible, deterministic gadget where each state admits a single (opposite) transition between the locations and these transitions flip the state. It can be trivially simulated by taking a single tunnel of a locking 2-toggle or nondeterministic locking 2-toggle.

\begin{theorem} \label{thm:nlt}
	1-player planar motion planning with the \nlt is \PSPACE-complete.
\end{theorem}
\begin{proof}
	
	In the construction shown in Figure~\ref{fig:nl2t-to-l2t}, the agent can enter through either of the top lines; suppose they enter on the left. Other than backtracking, the agent's only path is across the bottom 1-toggle, then up the leftmost tunnel, having chosen the state of the \nlt which makes that tunnel traversable. 

	Now the only place the agent can usefully enter the construction is the leftmost line. The agent can only go down the leftmost tunnel, up the 1-toggle, and out the top right entrance, again making the appropriate nondeterministic choice when traversing the left gadget. 

	Symmetrically, if (from the unlocked state) the agent enters the top right, they must exit the bottom right, and the next traversal must go from the bottom right to the top right and return the construction to the unlocked state. Thus this construction simulates a locking 2-toggle.
\end{proof}

\begin{figure}
	\centering
	\includegraphics[scale=.8]{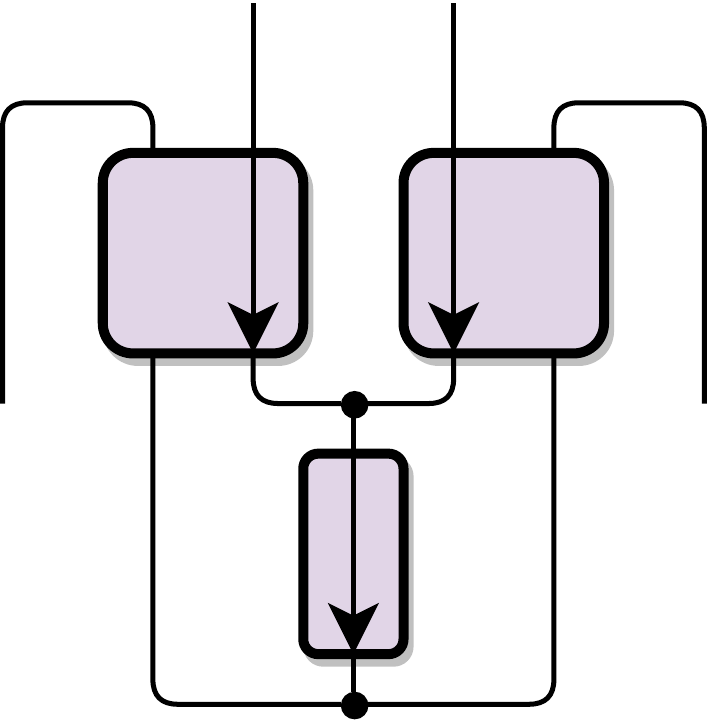}
	\caption{Constructing a locking 2-toggle from a nondeterministic locking 2-toggle. It is currently in the unlocked state. The nondeterministic locking 2-toggles are in leaf states (top states in Figure~\ref{fig:nl2t}).}
	\label{fig:nl2t-to-l2t}
\end{figure}

If we instead build the above construction with locking 2-toggles in leaf states, then all three of the locking 2-toggles used are in leaf states (the 1-toggle is one tunnel of a locking 2-toggle). A very similar argument as the \nlt construction shows this gadget also simulates a locking 2-toggle. Thus, given a 1-player motion planning problem with locking 2-toggles, we can replace all of the locking 2-toggles in nonleaf states with this gadget to obtain an instance where all starting gadgets are in leaf states.

\begin{corollary}
	1-player motion planning with the locking 2-toggle where all of the locking 2-toggles start in leaf states is \PSPACE-complete.
\end{corollary}

\later{
\section{3-port Self-Closing Door}
\label{app:self-closing door}

Ani et al.~\cite{doors} proved \PSPACE-completeness of 1-player planar motion planning with many types of self-closing door gadgets and all of their planar variations.
For completeness, we give a proof specific to the 3-port self-closing door gadget in this section.
Our proof is more succinct because it does not consider other variants of the gadget.
The reduction is from 1-player motion planning with the door gadget from \cite{nintendoor}. 

\begin{theorem}\label{thm:scd}
1-player planar motion planning with the 3-port self-closing door is \PSPACE-hard.
\end{theorem}
\begin{proof}
We will show that the 3-port self-closing door planarly simulates a crossover, which lets us ignore planarity.
We will then show that the 3-port self-closing door simulates the door gadget. Because 1-player motion planning with the door gadget is \PSPACE-hard \cite{nintendoor}, so is 1-player motion planning with the 3-port self-closing door, and because it simulates a crossover, so is 1-player planar motion planning with the 3-port self-closing door. Along the way, we will construct a self-closing door with multiple door and opening ports as well as a diode.

\paragraph{Diode.}We can simulate a diode (one-way tunnel which is always traversable)
by connecting the opening port to the input of the self-closing tunnel. The agent can always go to the opening port
and then through the self-closing tunnel, but can never go the other way because the self-closing tunnel is directed.

\paragraph{Port Duplicator.} The construction shown in Figure~\ref{fig:scd-port-duplicator} simulates a self-closing door with two equivalent opening ports. If the agent enters from the top, it can
open only one of the upper gadgets, then open the lower gadget, and then must exit the same way it came. Note, this same idea can be used to construct more than two ports, which will be needed later.

\begin{figure}
	\centering
	\includegraphics[scale=.8]{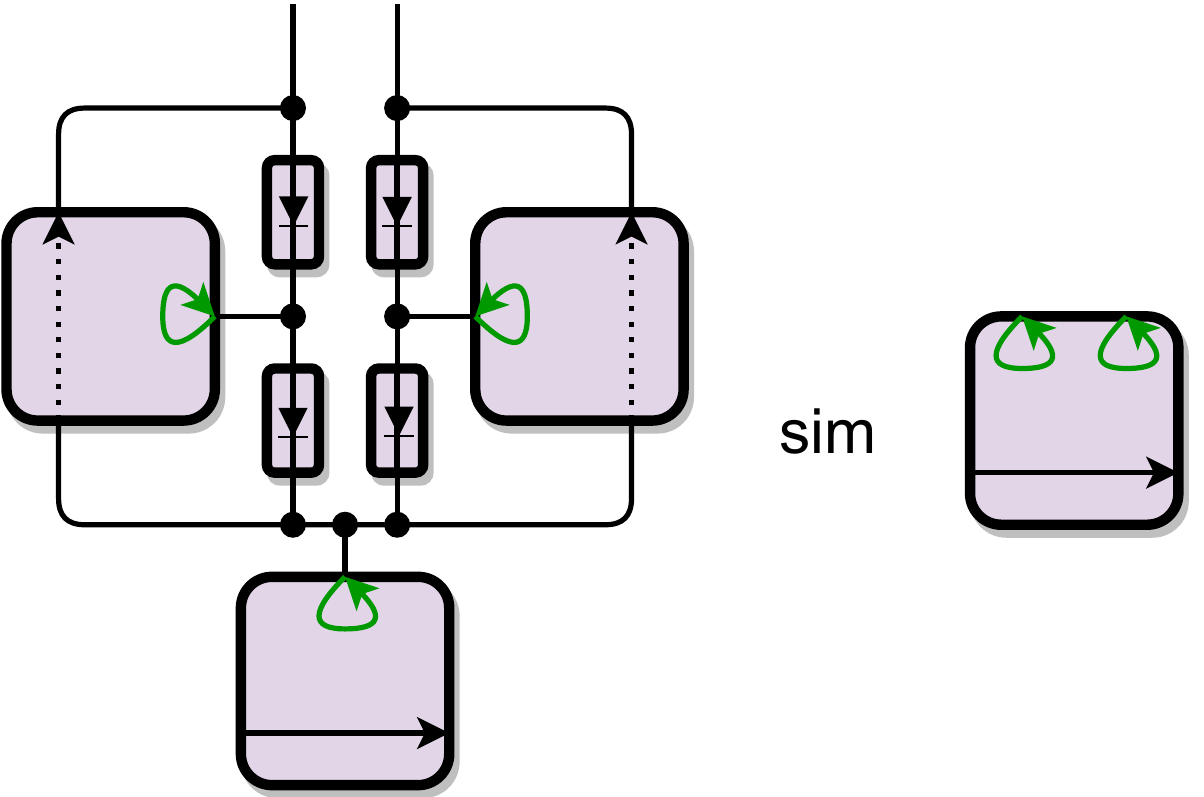}
	\caption{3-port self-closing door simulating a version of it that has 2 opening ports. Opening ports are shown in green.
	A dotted self-closing tunnel is closed, and a solid self-closing tunnel is open.}
	\label{fig:scd-port-duplicator}
\end{figure}

We use these to simulate an intermediate gadget composed of two of self-closing doors each connected to two opening ports in a particular order arrangement, shown in Figure~\ref{fig:planar-scd-crossoverish}.  If the agent enters from port 1 or 4,
it will open door E or F, respectively, and then leave. If the agent enters from port 2, it can open doors A, B, and C. If it then traverses door B and opens door E, it will get stuck because both B and D are closed. So the agent cannot open door E and exit.
Instead, it can traverse doors B and A, ending up back at port 2 with no change except that door C is open. Entering
port 2 or 3 always gives the agent an opportunity to open door C, so leaving door C open does not help.
So the only useful path after entering port 2 is to traverse door C. The agent is then forced to go right and can open door F. Then
it is forced to traverse door B. Again if the agent opens door E, it will be stuck, so the agent traverses door A instead and
returns to port 2, leaving door F open. 
Similarly, if the agent enters from port 3, the only useful thing it can do is open
door E and return to port 3.

\begin{figure}[t]
	\centering
	\includegraphics[scale=.8]{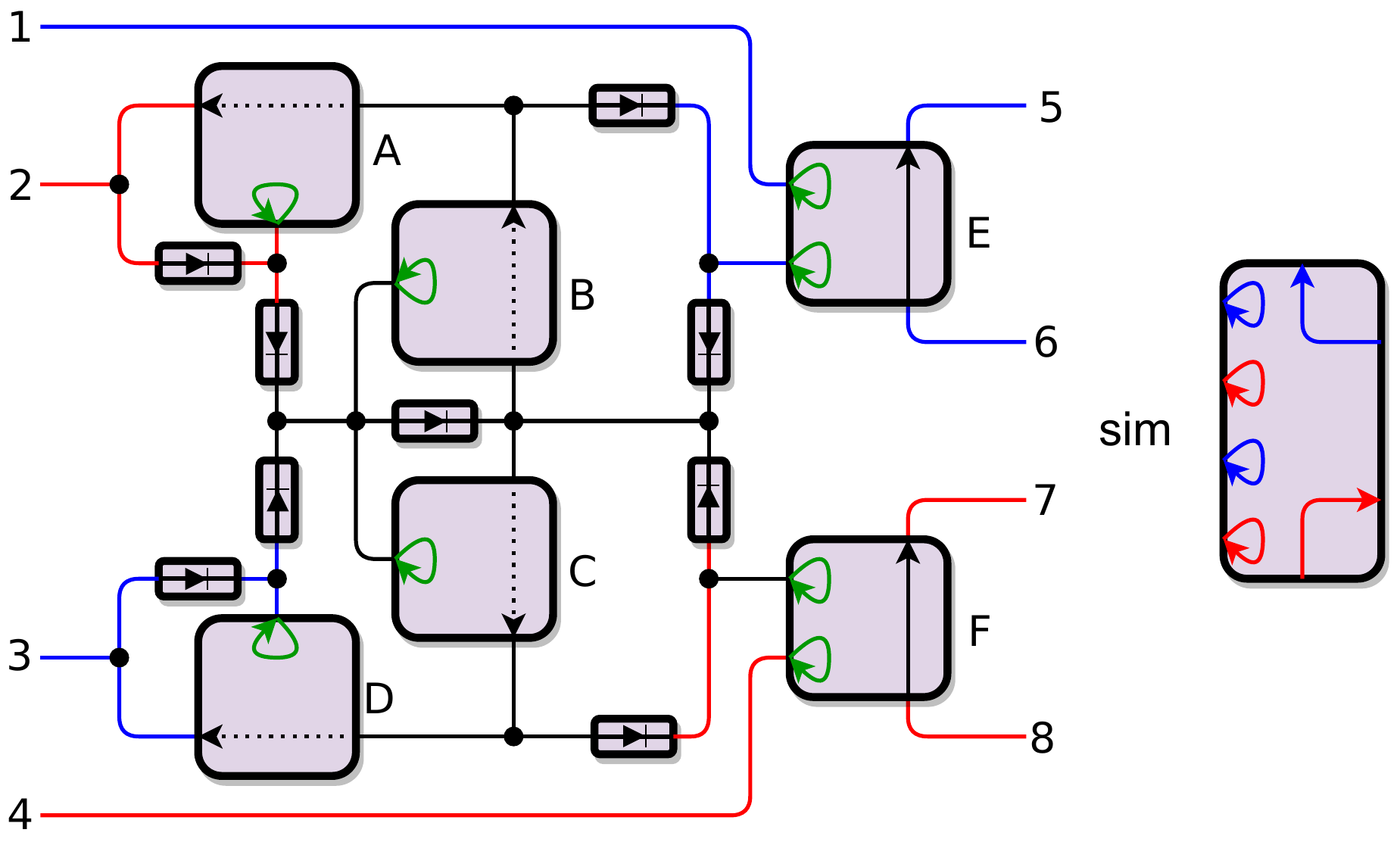}
	\caption{3-port self-closing door simulating the gadget on the right, where each port opens the door of the same color (the top and third-from-top open the top door, and the others open the bottom door).}
	\label{fig:planar-scd-crossoverish}
\end{figure}

\paragraph{Crossover.} This intermediate gadget can simulate a directed crossover, shown in Figure~\ref{fig:planar-scd-crossover}. If the agent enters at the top left, it can open the left door on the top gadget, open both doors on the bottom gadget, and then exit the bottom right while closing all three opened doors. If the agent opens both doors on the top gadget it will get stuck. Similarly if the agent enters the bottom left, all it can do is exit the top right.
The directed crossover can simulate an undirected crossover, as in
Figure~\ref{fig:dir-crossover} and shown in \cite{Push100}.

\begin{figure}
	\centering
	\includegraphics[scale=.8]{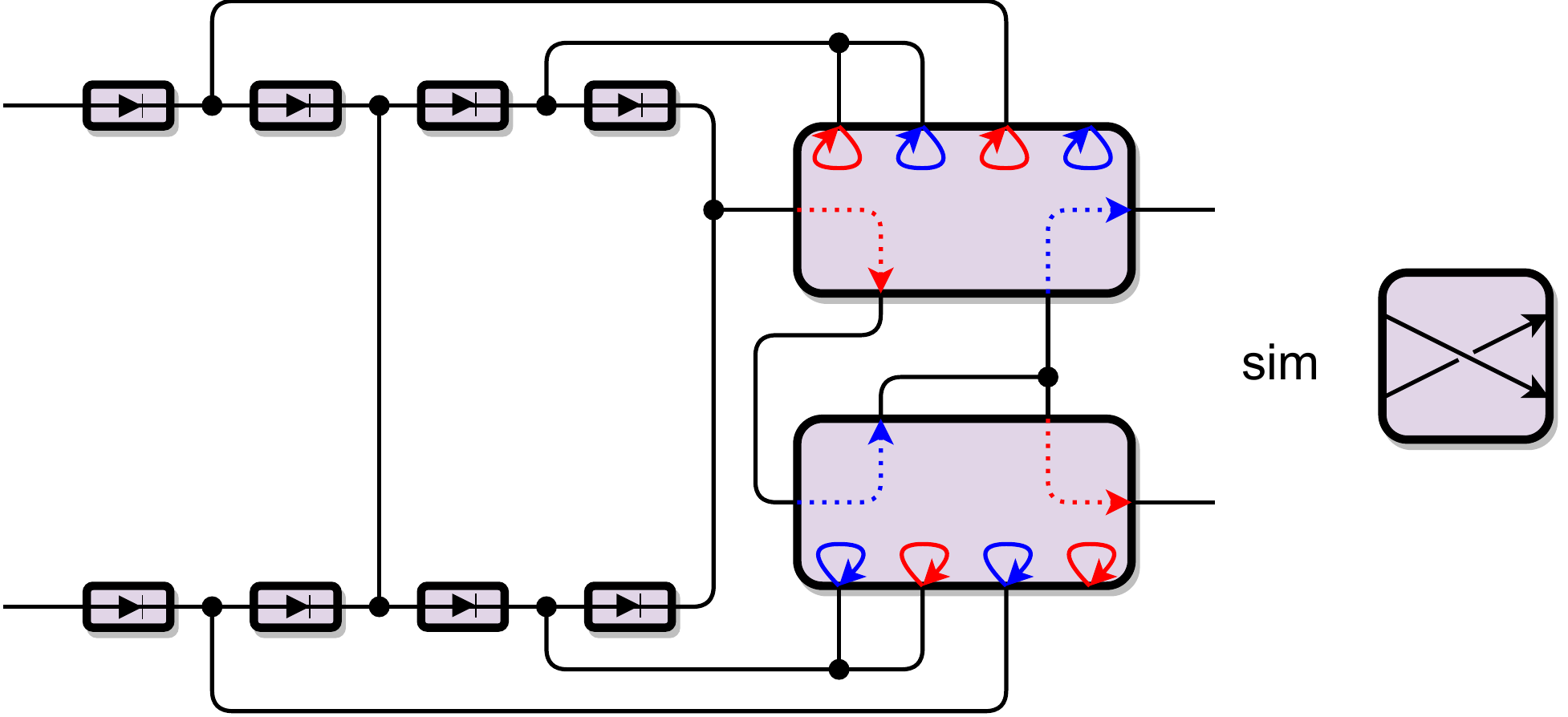}
	\caption{3-port self-closing door simulating a crossover.}
	\label{fig:planar-scd-crossover}
\end{figure}

\begin{figure}
	\centering
	\includegraphics[scale=.8]{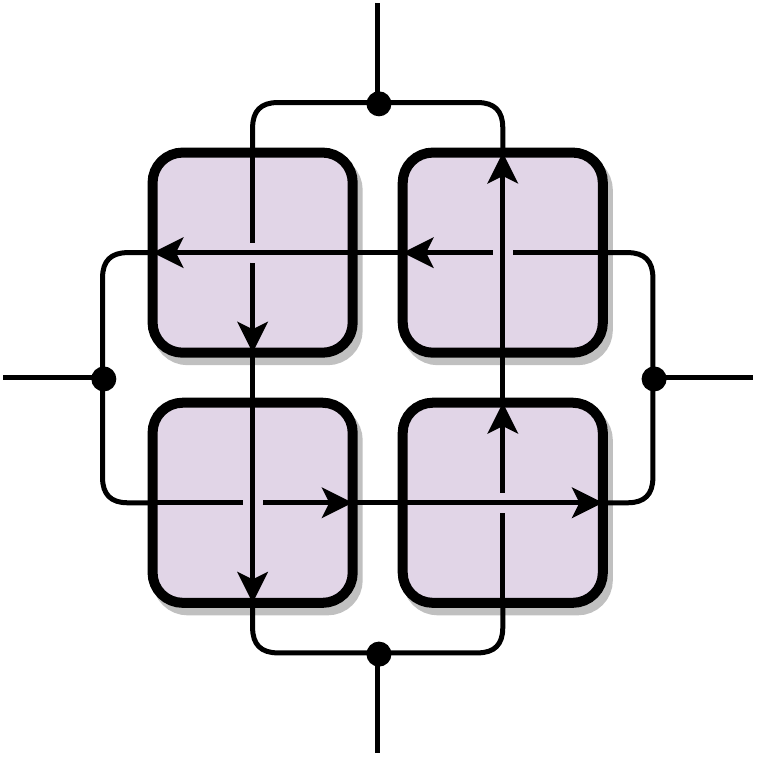}
	\caption{Directed crossover simulating an undirected crossover.}
	\label{fig:dir-crossover}
\end{figure}

\paragraph{Door Duplicator.} Now, we use this crossover to simulate a gadget with two self-closing doors controlled by the same opening port, as shown in Figure~\ref{fig:scd-tunnel-duplicator}. This gadget has two states, open and closed. Both doors are either open or closed and going through either door closes both of them.
The construction is similar to the construction for the port duplicator, but goes through a tunnel instead.

\begin{figure}
	\centering
	\includegraphics[scale=.8]{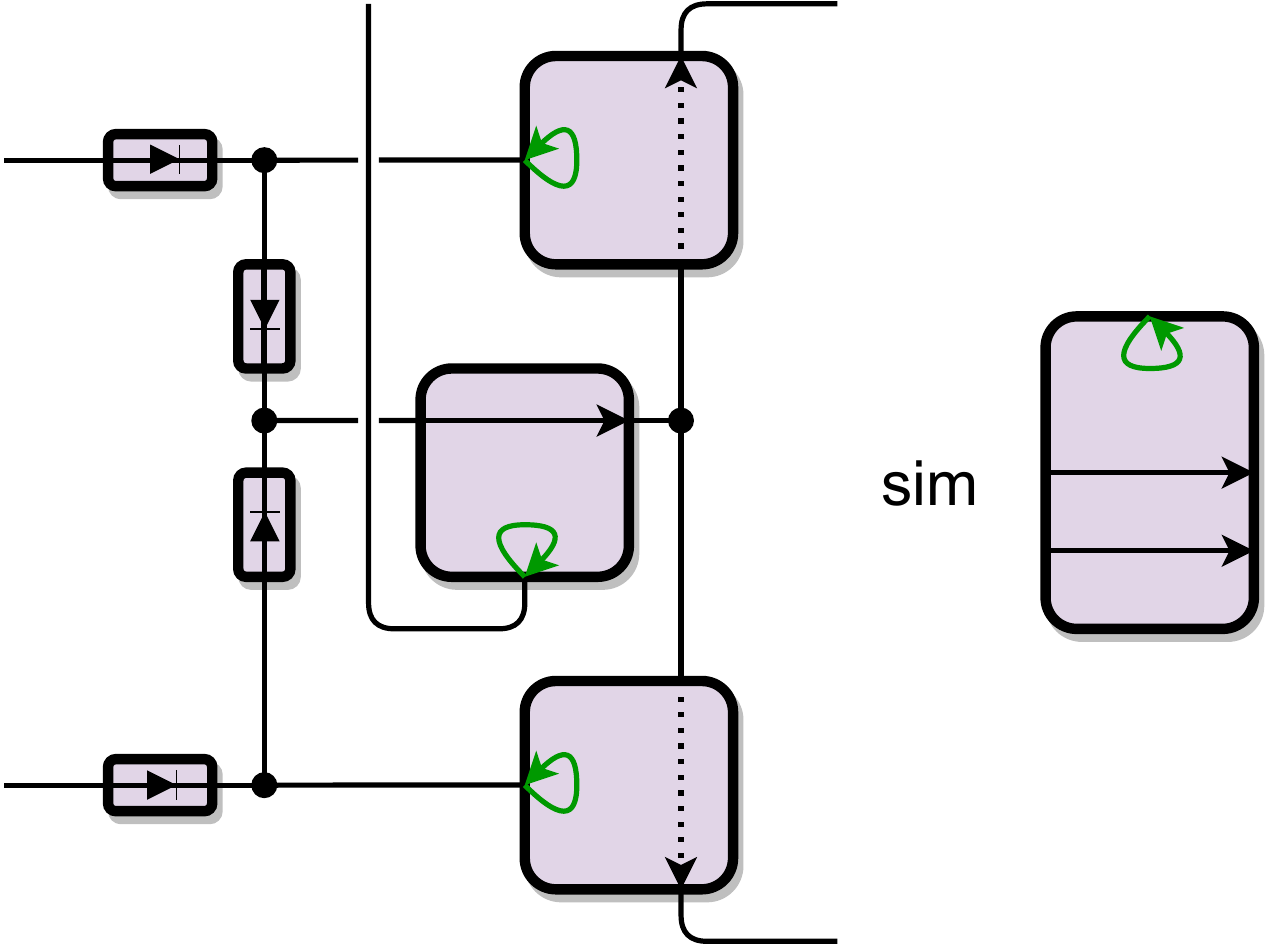}
	\caption{3-port self-closing door simulating a gadget with 2 self-closing tunnels.}
	\label{fig:scd-tunnel-duplicator}
\end{figure}

\paragraph{Door Gadget.} Finally, we triplicate the opening port by adding a third entrance to the construction in Figure~\ref{fig:scd-port-duplicator} similar to the other two, and use these ports to simulate a door gadget as shown
in Figure~\ref{fig:scd-otc}. Recall the whole three-port two-door gadget has only two states, open and closed. The agent can open both doors from any of the open ports and going across either self-closing door will close both doors. If the agent enters from port $O$, it can open the doors and leave.
If the agent enters from port $T_0$ and the gadget is open, the agent can traverse the door and then reopen it using the third port. The agent then leaves at port $T_1$. If the agent enters from port $C_0$, it can open the gadget and then must traverse the bottom tunnel and leave at port $C_1$, closing the
gadget. 
\end{proof}

\begin{figure}
	\centering
	\includegraphics[scale=.8]{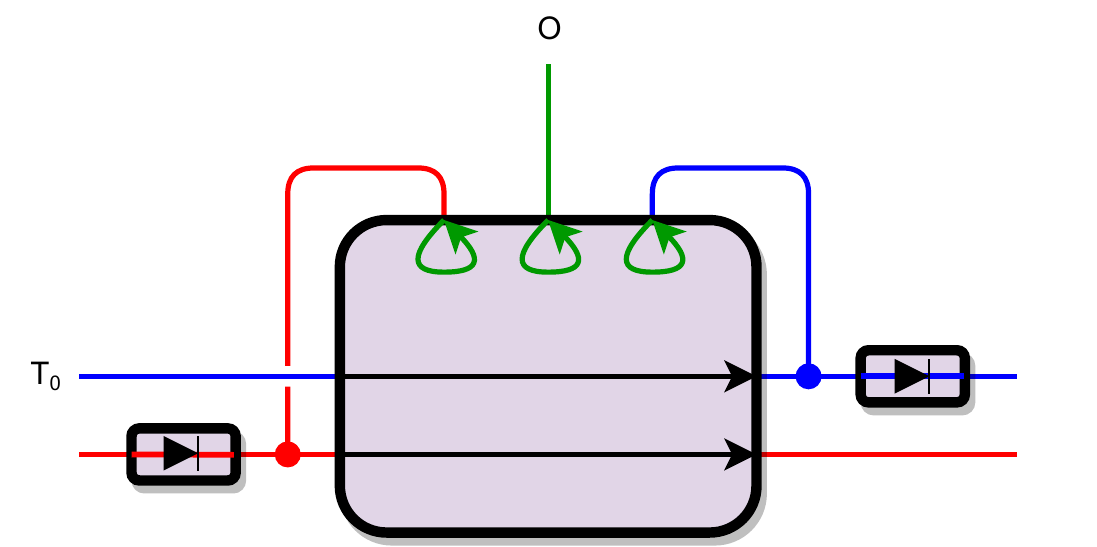}
	\caption{Simulation of the door gadget in \cite{nintendoor} using a gadget with 3 opening ports and 2 self-closing tunnels.}
	\label{fig:scd-otc}
\end{figure}
}

\subsection{\texorpdfstring{\PullkFG}{Pull?-kFG}}
\label{sec:optional pull}

In this section, we show that several versions of pulling-block problems with optional pulling and gravity are \PSPACE-complete by a reduction from 1-player motion planning with nondeterministic locking 2-toggles, shown \PSPACE-hard in Section~\ref{sec:nondet2toggle}.

We begin with a construction of a 1-toggle, and then use those and an intermediate construction to build a nondeterministic 2 toggle.

\paragraph{1-toggle.}

A \emph{1-toggle} is a gadget with a single tunnel, traversable in one direction.  When the agent traverses it, the direction that it can be traversed is flipped, meaning that the agent must backtrack and return the way it came in order to be able to traverse it the first way again.

Our 1-toggle construction in \PullkFG for $k\geq2$ is shown in Figure~\ref{fig:1toggle}.  In the state shown, it can only be traversed from left to right by pulling both blocks to the left.  This traversal flips the direction that the gadget can be traversed---it can now only be traversed from right to left.

\begin{figure}
	\centering
	\includegraphics[width=0.2\textwidth]{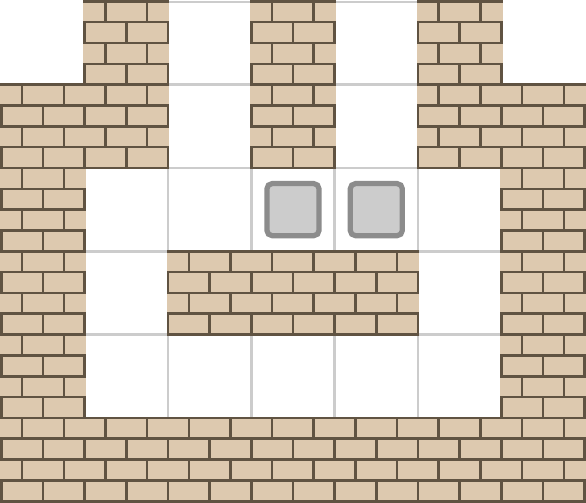}
	\caption{1-toggle in \PullkFG[2].}
	\label{fig:1toggle}
\end{figure}

\paragraph{Nondeterministic Locking 2-toggle.}

Our construction of a \nlt, shown in Figure~\ref{fig:locking2toggle}, uses two 1-toggles plus a connecting section at the top.

 \begin{figure}[ht]
	\begin{minipage}[b]{0.55\linewidth}
		\centering
	\includegraphics[width=0.9\textwidth]{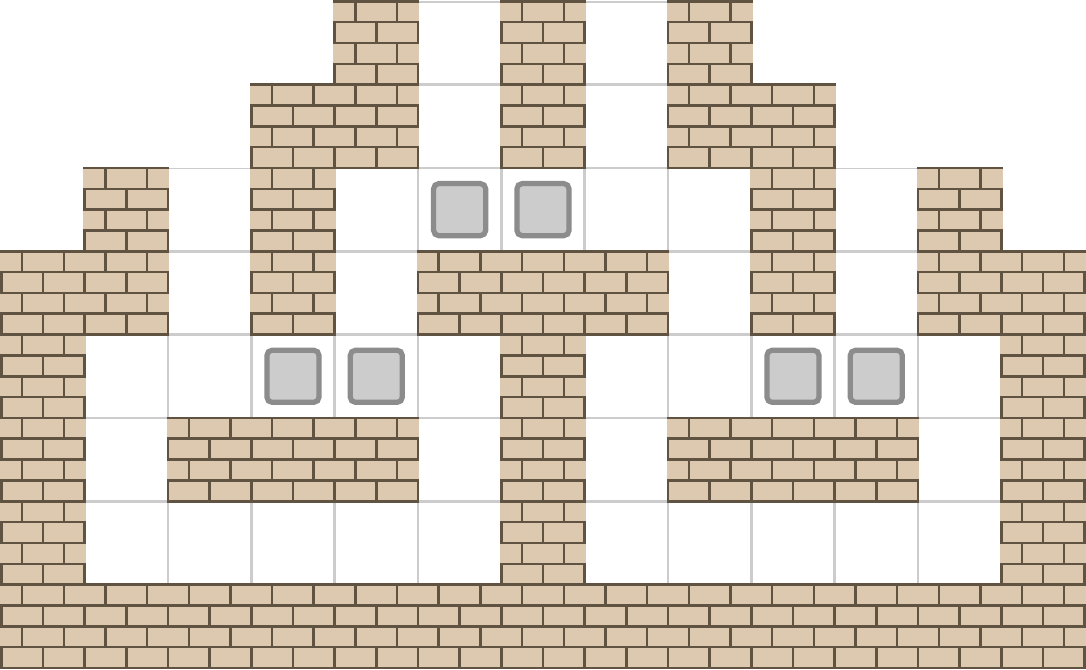}
	\caption{Locking 2-toggle in \PullkFG[2].}
	\label{fig:locking2toggle}
	\end{minipage}
	\hspace{0.3cm}
	\begin{minipage}[b]{0.42\linewidth}
		\centering
	\includegraphics[width=0.9\textwidth]{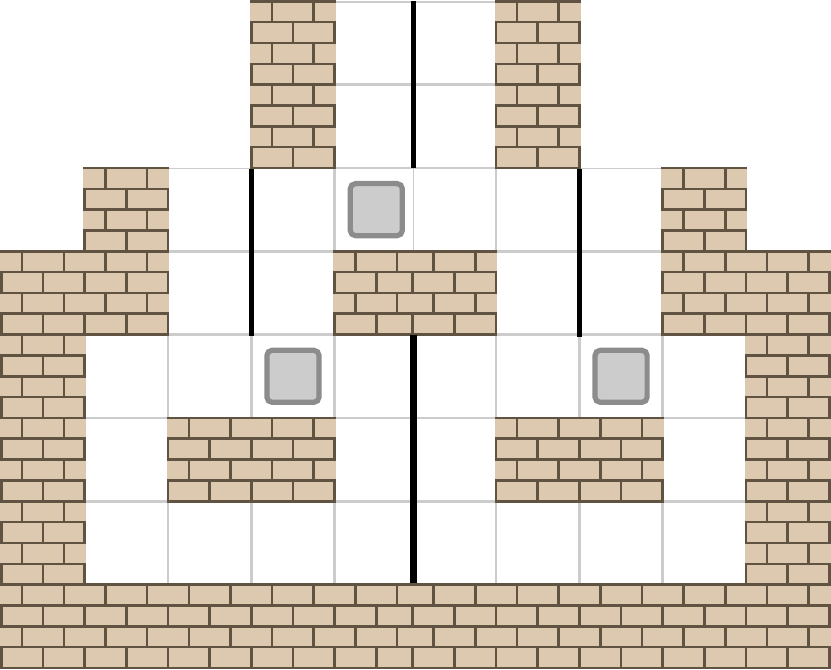}
	\caption{Locking 2-toggle in \PullkWG[1].}
	\label{fig:locking2toggle-W}
	\end{minipage}
\end{figure}

The configuration shown in Figure~\ref{fig:locking2toggle} is a leaf state. The right tunnel is traversable from to right to bottom right.  If the agent traverses that tunnel, it can choose whether to pull the top pair of blocks to the right (because pulling is optional), corresponding to the nondeterministic choice in the \nlt. Both 1-toggles will be in the state where they can be traversed from bottom (outside) to top (inside). One of these paths will be blocked by the top pair of blocks and the other will be traversable, depending on whether the agent chose to pull those blocks. Traversing the traversable path then puts the gadget in a leaf state, either the one shown or its reflection.

It is possible for the agent to pull only one block instead of two, but this can only prevent future traversals, so never benefits the agent.

\begin{theorem}
\label{thm:pull-kFG-PSPACE-complete}
\PullkFG is \PSPACE-complete for $k \ge 2$ and $k=*$.
\end{theorem}

\begin{proof}
Lemma~\ref{lem:pull-in-PSPACE} gives containment in \PSPACE.
For hardness, we reduce from 1-player planar motion planning with the \nlt, shown \PSPACE-hard in Theorem~\ref{thm:nlt}. We embed any planar network of gadgets in a grid, and replace each \nlt with the construction described above in the appropriate state. The resulting pulling-block problem is solvable if and only if the motion planning problem is.

This reduction works for \PullkFG for any $k \ge 2$ including $k=*$, because the player only ever has the opportunity to pull 2 blocks at a time. This proof requires optional pulling because the player must choose whether to pull blocks while traversing a \nlt.
\end{proof}

\begin{corollary}
\PullkWG is \PSPACE-complete for $k\ge1$ and $k=*$.
\end{corollary}

\begin{proof}
With thin walls, the tunnels can be separated by a thin wall instead of a fixed block, which means that only one block is required in each of the toggles. This is shown in Figure~\ref{fig:locking2toggle-W}.  The rest of the proof follows in the same manner, demonstrating \PSPACE-completeness of \PullkWG for $k \ge 1$.
\end{proof}

\subsection{\texorpdfstring{\PullkFG[$k$][!]}{Pull!-kFG}}
\label{sec:mandatory gravity}

In this section, we show \PSPACE-completeness for pulling-block problems with forced pulling and gravity, using a reduction from 1-player planar motion planning with the 3-port self-closing door, shown \PSPACE-hard in Theorem~\ref{thm:scd}.

\begin{theorem}
\label{thm:pull!-kFG-PSPACE-complete}
\PullkFG[$k$][!] is \PSPACE-complete for $k \ge 1$ and $k=*$.
\end{theorem}

\begin{proof}
Lemma~\ref{lem:pull-in-PSPACE} gives containment in \PSPACE.
We show \PSPACE-hardness by a reduction from 1-player planar motion planning with the 3-port self-closing door. It suffices to construct a 3-port self-closing door in \PullkFG[$k$][!].

First, we construct
a diode, shown in Figure~\ref{fig:pull!-kFG-diode}. The agent cannot enter from the right. If the agent enters from the left,
it must pull the left block to the left to advance. If it pulls the left block left and then exits, they still cannot enter from the right,
so doing so is useless. The agent then advances and is forced to pull the left block back to its original position. The agent then must
pull the right block left to advance, and must actually advance because the way back is blocked. As the agent exits the gadget, it
is forced to pull the right block back to its original position. Therefore, the agent can always cross the gadget from left to right
and never from right to left, simulating a diode.

\begin{figure}
\centering
\includegraphics[width=0.5\textwidth]{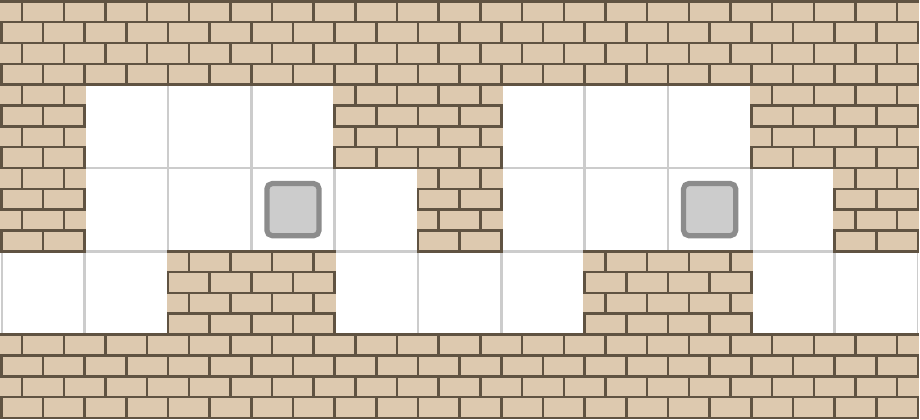}
\caption{A diode in \PullkFG[$k$][!].}
\label{fig:pull!-kFG-diode}
\end{figure}

Using this diode, we then construct a 3-port self-closing door, shown in Figure~\ref{fig:pull!-kFG-scd}; the diode icons indicate the diode shown in Figure~\ref{fig:pull!-kFG-diode}. The bottom is exit-only. In the closed
state, the agent should not enter from the top because it would become trapped between a block and the wrong end of a diode. The
agent can enter from the right, pull the block 1 tile right, and leave, opening the gadget. In the open state, the agent can enter
from the top and exit out the bottom, and is forced to pull the block back to its original position, closing the gadget. So this construction
simulates a 3-port self-closing door.

\begin{figure}
\centering
\subfloat[closed]{\includegraphics[width=0.35\textwidth]{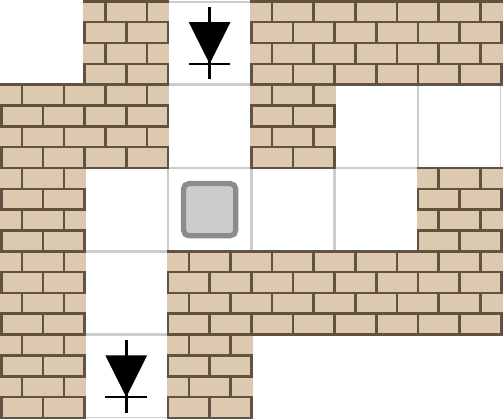}}
\hfill
\subfloat[Open]{\includegraphics[width=0.35\textwidth]{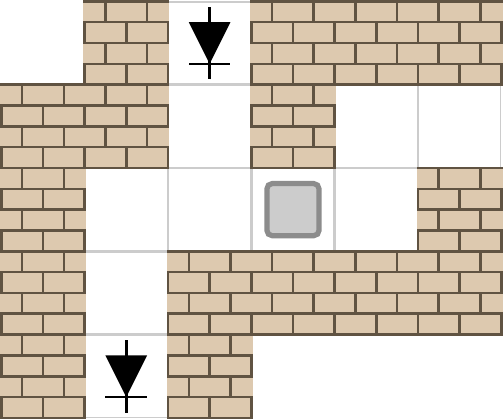}}
\caption{A 3-port self-closing door in \PullkFG[$k$][!].}
\label{fig:pull!-kFG-scd}
\end{figure}

Because the player never has the opportunity to pull multiple blocks, this reduction works for all $k\ge1$ including $k=*$.
\end{proof}

\section{\texorpdfstring{\PullkFG[1]}{Pull?-1FG} is NP-hard}
\label{sec:Pull1FG NP}

In this section, we show \NP-hardness for \PullkFG[1] by reducing from
1-player planar motion planning with the crossing NAND gadget from \cite{doors}.
A \emph{crossing NAND gadget} is a three-state gadget with two crossing
tunnels, where traversing either tunnel permanently closes the other tunnel. 
1-player planar motion planning with the crossing NAND gadget is NP-hard
in \cite[Lemma~4.9]{doors} based on the constructions in
\cite{demaine2003pushing,friedman2002pushing} which originally reduce from \PTC.

\begin{theorem}
	\label{thm:pull-1FG-NP-hard}
	\PullkFG[1] is \NP-hard.
\end{theorem}

\begin{proof}

We reduce from 1-player planar motion planning with the crossing NAND gadget
\cite[Lemma~4.9]{doors}.
First we first construct a ``single-use'' one-way gadget, shown in
Figure~\ref{fig:single-one-way}.
This gadget can initially can be crossed in one way, but then becomes
impassable in both directions.

\begin{figure}
\centering
$\vcenter{\hbox{\includegraphics[width=0.3\textwidth]{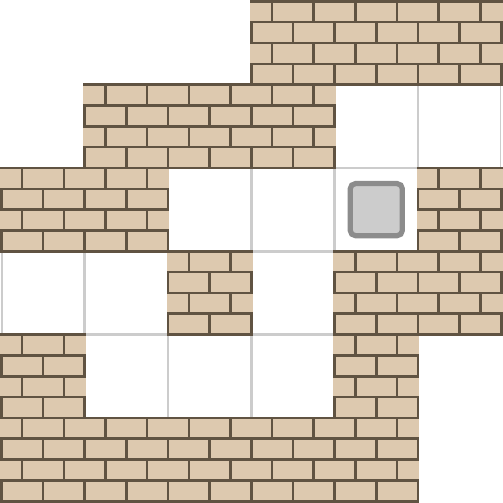}}}
\quad \equiv \quad
\vcenter{\hbox{\includegraphics{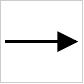}}}$
\caption{Single-use one-way gadget that initially allows traversal from left-to-right and then
    prevents traversal in both directions.}
\label{fig:single-one-way}
\end{figure}

Figure~\ref{fig:crossover} shows our construction of the crossing NAND gadget.
Single-use one-way gadgets enforce that the agent must enter
through one of the top paths.
The agent must pull two blocks to enter the gadget;
these blocks end up stacked in the vertical tunnel on top of the block below.
The agent cannot exit via the bottom tunnel underneath its entry tunnel:
the agent can pull one block into the slot on the bottom, and then can pull
one block one square, but that still leaves the third block of the stack
blocking off the exit path.
The agent cannot exit via the other top path, because it is blocked by the
single-use one-way gadget.
The only path remaining is for the agent to cross diagonally by pulling the
single block in the lower layer into the slot, revealing a path to the exit
opposite where the agent entered.  
After leaving, both the entry tunnel and exit tunnel are impassable
because the single-use one-way gadgets have become impassable.
If the agent later enters via the other entry tunnel, the agent will be trapped,
because it will not be able to leave via the tunnel that was ``collapsed''
in the initial entry.
\end{proof}

\begin{figure}
\centering
\includegraphics[width=0.55\textwidth]{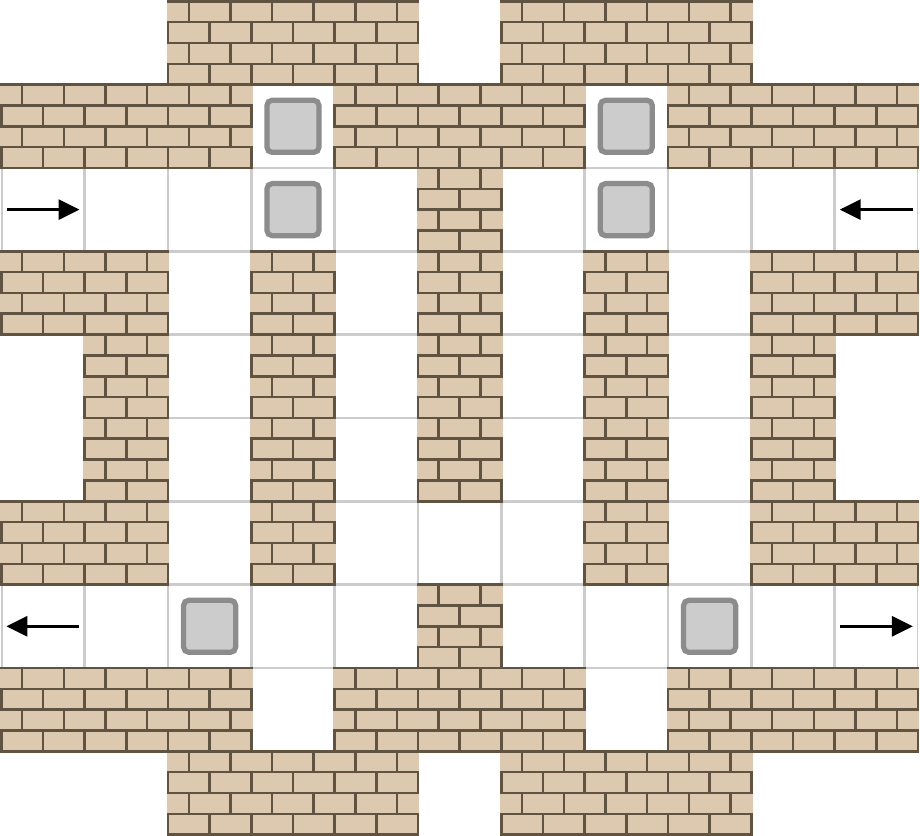}
\caption{Crossing NAND gadget allowing traversal either from the top-left to
  the bottom-right, or from the top-right to the bottom-left.  After being
  traversed once, the entire gadget becomes impassable in any direction.}
\label{fig:crossover}
\end{figure}

We leave open the question of whether \PullkFG[1] is in \NP or \PSPACE-hard.

\section{Open Problems}
\label{sec:Open Problems}
There are several open problems remaining related to the pulling-block problems considered in this paper.
\begin{enumerate}
  \item What is the complexity of \PullkFG[1] (the last remaining problem in Table~\ref{tab:results})? We leave a gap between \NP-hardness and containment in \PSPACE.
  
  \item What is the complexity of pulling-block puzzles without fixed blocks (say, on a rectangular board)? With block pushing, one can generally construct effectively fixed blocks by putting enough blocks together. This technique no longer works in the block-pulling context.
  
  \item Do all of these variants remain \PSPACE-hard when we ask about storage (can the player place blocks covering some set of squares?)\ or reconfiguration (where blocks are distinguishable and must reach a desired configuration) instead of reachability? The storage question for \PullkFG[$k$][?] for $k \geq 1$ and \PullkFG[$*$][?] has been proved PSPACE-hard \cite{PRB16}.
  
  \item What about the studied variants applied to \PushPull (where blocks can be pushed and pulled) and \PullPull (where blocks must be pulled maximally until the robot backs against another block)?  Standard versions are proved PSPACE-complete in \cite{demaine2017push,PRB16}, but variations with mandatory pulling, gravity, and/or no fixed blocks all remain open.
\end{enumerate}

\bibliographystyle{alpha}
\bibliography{thebib}

\appendix
\magicappendix

\end{document}